\documentclass[journal,draftcls,onecolumn,12pt,twoside]{IEEEtran}
\IEEEoverridecommandlockouts 
\usepackage{cite}
\usepackage{amsfonts}
\usepackage[normalem]{ulem}
\usepackage{algorithm}
\usepackage{algorithmic}

\usepackage{soul}
\usepackage{comment}
\usepackage{amssymb, amsthm, amsmath, graphicx, enumerate, cases}
\usepackage{array}
\usepackage{changepage}
\usepackage{xcolor}
\usepackage{balance} 
\usepackage{multirow,booktabs}

\DeclareMathOperator{\lcm}{lcm}

\newcommand{\Aa}{\color{cyan}{a_0}}
\newcommand{\Ab}{\color{cyan}{a_1}}
\newcommand{\Ac}{\color{cyan}{a_2}}
\newcommand{\Ad}{\color{cyan}{a_3}}
\newcommand{\Ae}{\color{cyan}{a_4}}
\newcommand{\Af}{\color{cyan}{a_5}}
\newcommand{\Ag}{\color{cyan}{a_6}}
\newcommand{\Ah}{\color{cyan}{a_7}}
\newcommand{\Ai}{\color{cyan}{a_8}}
\newcommand{\Aj}{\color{cyan}{a_9}}
\newcommand{\Ak}{\color{cyan}{a'_0}}
\newcommand{\Al}{\color{cyan}{a'_1}}
\newcommand{\Ba}{\color{blue}{b_0}}
\newcommand{\Bb}{\color{blue}{b_1}}
\newcommand{\Bc}{\color{blue}{b_2}}
\newcommand{\Bd}{\color{blue}{b_3}}
\newcommand{\Be}{\color{blue}{b_4}}
\newcommand{\Bf}{\color{blue}{b_5}}
\newcommand{\Bg}{\color{blue}{b_6}}
\newcommand{\Bh}{\color{blue}{b_7}}
\newcommand{\Ca}{\color{violet}{c_0}}
\newcommand{\Cb}{\color{violet}{c_1}}
\newcommand{\Cc}{\color{violet}{c_2}}
\newcommand{\Cd}{\color{violet}{c_3}}
\newcommand{\Ce}{\color{violet}{c_4}}
\newcommand{\Cf}{\color{violet}{c_5}}
\newcommand{\Cg}{\color{violet}{c_6}}
\newcommand{\Ch}{\color{violet}{c_7}}
\newcommand{\Da}{\color{magenta}{d_0}}
\newcommand{\Db}{\color{magenta}{d_1}}
\newcommand{\Dc}{\color{magenta}{d_2}}
\newcommand{\Dd}{\color{magenta}{d_3}}
\newcommand{\De}{\color{magenta}{d_4}}
\newcommand{\Df}{\color{magenta}{d_5}}
\newcommand{\Dg}{\color{magenta}{d_6}}
\newcommand{\Dh}{\color{magenta}{d_7}}
\newcommand{\Di}{\color{magenta}{d_8}}
\newcommand{\Dj}{\color{magenta}{d_9}}
\newcommand{\Dk}{\color{magenta}{d'_{0}}}
\newcommand{\Dl}{\color{magenta}{d'_{1}}}
\newcommand{\mi}[1]{m_{#1}}

\newtheorem{theorem}{Theorem}
\newtheorem{lemma}{Lemma}
\newtheorem{remark}{Remark}
\newtheorem{proposition}{Proposition}
\newtheorem{corollary}{Corollary}
\newtheorem{definition}{Definition}

\newtheorem{example}{Example}

\begin{document}

\title{Color Multiset Codes based on Sunmao Construction
}

\author{
{Wing Shing Wong,~\IEEEmembership{IEEE Life Fellow},
Chung Shue Chen,~\IEEEmembership{IEEE Senior Member}, 
Yuan-Hsun Lo~\IEEEmembership{IEEE~Member}
}
\thanks{The research was partially funded by the National Science and Technology Council of Taiwan under Grant NSTC 114-2628-M-153-001-MY3. 
\textit{(Corresponding author: Yuan-Hsun Lo)}
}
\thanks{W. S. Wong is with the Department of Information Engineering, The Chinese University of Hong Kong (CUHK), Shatin, Hong Kong SAR. Email: wswong@ie.cuhk.edu.hk}
\thanks{C. S. Chen is with Nokia Bell Labs,  Paris-Saclay Center, 12 rue Jean Bart, Massy 91300, France. Email: chung\_shue.chen@nokia-bell-labs.com}
\thanks{Y.-H. Lo is with the Department of Applied Mathematics, National Pingtung University, Taiwan.  Email: yhlo0830@gmail.com}
}

\maketitle 

\begin{abstract}
We present results on coding using multisets instead of ordered sequences.  The study is motivated by a moving object tracking problem in a sensor network and can find applications in settings where the order of the symbols in a codeword cannot be maintained or observed.  In this paper a multiset coding scheme is proposed on source data that can be organized as a flat or cyclic multi-dimensional integer lattice (grid).   A 
fundamental idea in the solution approach is to decompose the original source data  grid into sub-grids.  The original multiset coding problem can then be restricted to each of the sub-grid.  Solutions for the sub-grids are subsequently piece together to form the desired solution.  We name this circle of idea as sunmao construction in reference to woodwork construction method with ancient origin.   Braid codes are specific solutions defined using the sunmao construction.  They are easy to define for multi-dimensional grids.  Moreover for a code of a given code set size and multiset cardinality, if we measure coding efficiency by the number of distinct symbols required, then braid codes have asymptotic order equal to those that are optimal.   We also show that braid codes have interesting inherent error correction properties.
\end{abstract}

\begin{IEEEkeywords}
Industrial IoT, object tracking, localization, multiset combinatorial coding. 
\end{IEEEkeywords}

\section{Introduction}
\label{section:introduction}
{\it Color multiset coding} is defined over a source symbol set arranged in an $n$-dimensional cyclic or flat integer lattice, exploiting the grid topology for coding and decoding.  It was introduced and analyzed in 
\cite{ISIT24,Paper1,ICALP24}.   This paper aims to provide a systematic
color multiset code construction approach based on a concept referred to here as {\it sunmao construction}.  The naming is
inspired by the abstract resemblance to a {\it sunmao} joint, also known as a mortise-and-tenon joint, in architecture with ancient origin.

First of all, we briefly recall the idea of color multiset coding for the readers' convenience.  
Let $G$ be an $n$-dimensional cyclic grid (integer lattice)  with size defined by ${\bf M}=(M_1,\ldots, M_n)$ so that the grid contains exactly $M_1\times \dots \times M_n$ elements.  Let $\bf j$ be a grid point, which can also be regarded as an element in $\mathbb{Z}^n$.  Given ${\bf m}=(m_1,\ldots,m_n) \leq {\bf M}$, define the ${\bf m}$-{\it block} tagged at grid point $\bf j$ by the set of neighboring points, $B_{\bf m}({\bf j}) \triangleq \{ {\bf k} ~mod~ {\bf M}:{\bf j \leq k < j+m}, {\bf k}\in\mathbb{Z}^n\}$. 
We refer to $\textbf{M}$ and $\textbf{m}$ as the {\it grid size} and {\it block size} respectively.  
The parameters and notation will be formally defined in Section~\ref{section:notation}.
Denote the volume of $\bf m$, $\prod_{i=1}^{n}m_i$, by $\nu({\bf m})$ .

Let $\cal{A}$ be an alphabet.   We recall that a multiset is a set in which an element is allowed to appear multiple times \cite{multiset89}.  Let ${\cal{A}}_{l}$ denote the multiset in which all elements of $\cal{A}$ appear  $l$ times.

Let $\Phi$ be a mapping from $G$ to $\cal{A}$. 
For presentation convenience, we refer the alphabet symbols as colors.  A color multiset code is defined by mapping a source symbol $\bf j$ to the multiset 
\[
\Phi(B_{\bf m}({\bf j})) = \{ \Phi({\bf k}): {\bf k} \in B_{\bf m}({\bf j}) \} \subset {\mathcal{A}_{\nu(\bf m)}}.
\]
Note that all codewords have cardinality $\nu({\bf m})$.

Such a code has an interesting property --- neighboring source symbols are mapped to codewords that only differ slightly.   
For example, when $G$ is a 1-dimensional (1D) cyclic grid of size $M$ with block size $m$, the codewords for the grid points $x$ and $x+1$ differ by one element if $\Phi(x)\ne \Phi(x+m ~{\rm mod}~ M)$, or are identical otherwise.  
It follows that color multiset coding offers a method to construct  {\it combinatorial Gray codes} \cite{CGC97} when the source set is imposed with a grid topology.

It is worth noting that a more traditional approach when $G$ is a 1D grid is to define the codeword at $x$ by the sequence, $(\Phi(x),\ldots,\Phi(x+m-1 ~{\rm mod}~ M))$.  
However, there are applications in which using multisets as codewords are natural.  Object tracking over a proximity sensor network provides one such example\cite{ISIT24}.

Consider a mobile object tracking problem over a cyclic $n$-dimensional grid, $G$.  We assume the grid contains grid points of the form ${\bf x}=(j_1\delta,..., j_n\delta)$, where $j_i$'s are integers satisfying $0 \leq j_i < M_i$ 
for $(M_1,\ldots, M_n)={\bf M}$.
At each time slot $t$, the tracked object can randomly appear at a grid point in $G$. Proximity sensors are placed at all grid points so that for a positive integer vector ${\bf m} =(m_1,...,m_n)$, the sensor located at $\bf x$ can detect the object if it is located in the area $\{{\bf z}=(z_1,..., z_n) : (l_i+m_i-0.5)\delta \leq z_i  < l_i+m_i+0.5)\delta, \forall i \}$. 
If each sensor is equipped with a color LED indicator, which is turned on when the object is detected, then to
a distant observer who can only resolve light frequency and intensity but not light position, the multisets of lit colors provide the sole data to locate the moving object.
In \cite{Paper1}, the resulting multiset code is referred to as a \textit{multiset combinatorial Gray code} (MCGC).  It is easy to see that the above coding problem can be reformulated to the current context in this paper,

In addressing the object tracking problem it is natural to search for solutions that require the minimal number of light colors.  In other words, a natural performance goal is to find color coding schemes that can guarantee each {\bf m}-block in $G$ maps to a unique color multiset by using the minimum number of color symbols.
For certain system parameters, non-singular codes with minimum number of  colors are known \cite{Paper1}, but for systems with general parameters, this is an open problem.

It is worth noting that if the codewords are defined as $m$-sequences, the de Bruijn sequences~\cite{de1946combinatorial,Etzion84}
can offer optimal solutions to the problem.   Hence, our work can be viewed as a generalization of Bruijn sequences.
The color multiset code is also related to the idea
of universal cycles~\cite{CDG92} and universal cycles for multisets~\cite{HJZ09}.
The concept of $M$-sequences \cite{Kumar92} also involves multisets, but the studied objects are nevertheless viewed as color sequences, similar to de Bruijn sequences.

Obviously, there are other approaches to define
MCGC on multi-dimensional grids.  For example, there are
some related results in the literature extending de Bruijn sequences beyond 1D grids.  These include the following:
A two-dimensional generalization of the binary de Bruijn sequences is the prefect map \cite{Etzion88,Paterson94}, which is  used in range-finding, data scrambling, position location \cite{Position-Sensing93,Bruckstein12} and image vision 
applications \cite{PAGES2005}. 
Higher-dimensional extension of prefect maps and de Bruijn sequences can also be defined. 
Construction methods of orientable sequences, that is, cyclic binary sequences in which sub-sequences of length $m$ occurs at most once in either direction can be found in \cite{Orientable-Sequences22}. 
However, these combinatorial objects do not directly lead to MCGC.
Methods and constructions to generate binary codes for correction of a multidimensional  clustered error through component code can also be found in \cite{Error-Correction09,Error-correction98}. 

\subsection*{Contribution and Organization}
\label{section:contribution}

Given a grid $G$ and block size $\bf M$, our goal is to find efficient non-singular multiset color coding algorithms on $G$.  Here, efficiency is not measured by the absolute number of colors required.  Instead. for a given block size, we analyze the asymptotic order of the number of colors required as the grid volume, $\nu({\bf M})=\prod_{i=1}^{n} M_i$, tends to infinity and adopt this as the performance measure in this paper.

The main contributions of the paper are summarized below:
\begin{itemize} 
    \item We introduce a recursive approach, known as {\it sunmao construction}, to synthesis multiset color code by first decomposing a given $\bf m$-block into component sub-blocks and decompose grid $G$ accordingly to sub-grids.   The final solution is then obtained by piecing color multiset codes on the sub-grids together.
    \item We introduce a class of non-singular multiset color
codes for 1D grids based on the sunmao construction.  The resulting codes are known as {\it braid codes}.
    \item When a sunmao construction is built on sub-blocks that are all singletons, we refer the construction as {\it unitary}.  Using unitary sunmao construction and the concept of {\it product code}, which was introduced in \cite{Paper1}, we generalize braid code from 1D to multi-dimensional grids.
    \item  We show that multi-dimensional unitary braid codes are non-singular and their required number of colors achieves the asymptotically minimal order as the grid size tends to infinity,
    \item We show that there is a fast algorithm for decoding braid codes.
    \item We show that 1D braid codes possess natural error correction properties.

\end{itemize}

The rest of the paper is organized as follows.
In Section~\ref{section:system-model}, we introduce the mathematical system model analyzed in this paper.  
The sunmao construction approach is explained in section~\ref{section:MT}.  
Braid codes for 1D grids are defined in Section~\ref{section:Braid 1}.
These constructions are extended to general $n$-dimensional grids in Section~\ref{section:Braid n}.  
The asymptotic order of the number of colors required to construct a non-singular unitary braid code is analyzed in Section~\ref{section:order}. 
We present a decoding algorithm for braid codes in Section~\ref{section:decoding} that does not rely on a lookup table.
We present some basic error correction properties of 1D braid codes in Section~\ref{section:non-product}.
Finally, we conclude the paper in Section~\ref{section:conclusion}. 

\section{System Model} 
\label{section:system-model}

\subsection{Mathematical Notation}
\label{section:notation}
Let $S$ be a multiset with distinct elements, $1, \ldots, k$,
and denote the multiplicity of element $i$ by $l_i$. 
We can also represent $S$ by $(l_1,\ldots,l_k)$.   We use bold fonts to represent a vector or matrix. 
For ${\bf x}=(x_1,\ldots,x_n), {\bf y}=(y_1,\ldots,y_n)$, $\bf x \leq y$ is equivalent to $x_i\leq y_i$ for all $i$ and
${\bf x}< {\bf y}$ if $x_i<y_i$ for all $i$.  
To highlight coordinate component $i$ in a vector, $\bf x$, we  represent it as $(x_i;{\bf x}^-_i)$, where ${\bf x}^-_i$ is the $(n-1)$-dimensional vector with the $i$-th component removed.
For example if ${\bf x}=(x_1, x_2, x_3,x_4)$, then ${\bf x}^-_3=(x_1,x_2,x_4)$ and we can represent $\bf x$ as $(x_3;{\bf x}^-_3)$.
${\bf 1}_n$ and ${\bf 0}_n$ denote the $n$-dimensional unit vector and zero vector respectively.   
  
Let $\mathbb{Z}^+$ denote the set of all positive integers.
For $n\in\mathbb{Z}^+$, let $\mathcal{Z}_n$ be the set consisting of $\{0,1,\ldots,n-1\}$ and $\mathbb{Z}_n=\{0,1,\ldots,n-1\}$ denote the ring of residues modulo $n$.
Note that $\mathcal{Z}_n$ differs from $\mathbb{Z}_n$: $\mathcal{Z}_n$ simply denotes the set of elements of $\mathbb{Z}_n$, without any ring structure implied.

\subsection{Grids and Product Grids}
\label{section:math-definition}
Given a positive integer $n$ and ${\bf M}=(M_1,...,M_n)$ with $M_1,...,M_n \in\mathbb{Z}^+$, define an $n$-dimensional integer lattice by $G_{\bf M} \triangleq \mathcal{Z}_{M_1}\times \cdots \times \mathcal{Z}_{M_n}=\{{\bf j}=(j_1,...,j_n):\,j_i\in \mathcal{Z}_{M_i} \}$,  and define an $n$-dimensional cyclic integer lattice by  $G^c_{\bf M} \triangleq \mathbb{Z}_{M_1}\times \cdots \times \mathbb{Z}_{M_n}=\{{\bf j}=(j_1,...,j_n):\,j_i\in\mathbb{Z}_{M_i} \}$.
The expression ${\bf j}\bmod {\bf M}$ stands for $j_i \bmod{M_i}$ for all $i$,
so $G^c_{\bf M}$ can also be defined as $\{{\bf j}\bmod {\bf M}:\, j_i\in\mathbb{Z}\}$. 
For simplicity, we refer to $G_{\bf M}$ as a \textit{flat grid} and $G^c_{\bf M}$ as a \textit{cyclic grid}.  

Given cyclic grids, $G^c_{{\bf M}_1}=\mathbb{Z}_{M_{1,1}}\times \cdots \times \mathbb{Z}_{M_{1,n_1}}$ and $G^c_{{\bf M}_2}=\mathbb{Z}_{M_{2,1}}\times \cdots \times \mathbb{Z}_{M_{2,n_2}}$, of dimension $n_1$ and $n_2$ respectively, the {\it product grid} of $G^c_{{\bf M}_1}$ and $G^c_{{\bf M}_2}$, $G^c_{{\bf M}_1} \times G^c_{{\bf M}_2}$, of dimension $n=n_1+n_2$ is defined to be: $\{(z_1,...,z_n):  z_i \in \mathbb{Z}_{M_{1,i}} {\rm ~if~} 1 \leq i \leq n_1, z_i \in \mathbb{Z}_{M_{2,i}} {\rm ~if~} n_1+1 \leq i \leq n \}.$
Product grid for flat grids can be defined similarly.

\begin{definition}
For ${\bf m}=(m_1,...,m_n) < {\bf M}$, an $\bf m$-$\it block$ at a grid point $\bf x$ of a cyclic grid $G^c_{\bf M}$ is a subset consisting of points
$\{ {\bf j} \mod {\bf M}: {\bf x \leq j < x+m}, {\bf j}\in \mathbb{Z}^n \}$. 
An $\bf m$-$\it block$ at $\bf x$ of a flat grid, which is defined only for ${\bf 0}_n \leq \bf x < M-m$, is the subset $\{ \bf j:x \leq j < x+m , {\bf j}\in \mathbb{Z}^n\}$. 
\end{definition}

Denote an $\bf m$-block by $B_{\bf m}( \bf x)$ or simply as $B(\bf x)$ or $B$ if there is no ambiguity. 
The volume of an $\bf m$-block is $v({\bf m})=\prod_{i=1}^{n}m_i$.
We refer to the set of grid points at which an $\bf m$-block is well defined as a \textit{coding area}.
Note that the coding area for $G^c_{\bf M}$ is identical to the grid.
For a flat grid, it is easy to see that all grid points in the coding area can be uniquely identified by $\bf m$-blocks.  
As for a cyclic grid, except for the trivial case $m_i= M_i$ for some $i$, all elements in $G^c_{\bf M}$ can also be identified by the ${\bf m}$-blocks they are tagged at.

\subsection{Color Multiset Code}

We formalize the coding algorithm, which is referred to as
{\it color multiset coding} or color coding for short, by a mapping, $\Phi$, from $G^c_{\bf M}$ (or $G_{\bf M}$) to a set of $k$ symbols. 
We refer to the symbols as colors and use $[k]\triangleq\{0,1,\ldots,k-1\}$ to represent them if there is no ambiguity. 
We call $\Phi$ a color mapping and use $\mathcal{C}_{{\bf M};k}$ to denote the
collection of color mappings on $G^c_{\bf M}$ (or $G_{\bf M}$) with $k$ colors.

\begin{definition}
For a given color mapping $\Phi\in\mathcal{C}_{{\bf M};k}$, the \textit{color code}word of a grid point, $\bf x$, is the multiset $\{\Phi({\bf j}): {\bf j} \in B_{\bf m}({\bf x})\}$.
It is $\bf m$-\textit{distinguishable} if all color codewords in the coding area are uniquely defined.
\end{definition}

A color mapping problem aims to find a distinguishable color code for a given grid and block size.   
In the rest of this paper, we focus our attention on cyclic grids. Most of the results on cyclic grids can be extended to flat grids with minor adaptation.
Denote by $K^c_{\bf M}(\bf m)$ the minimal number of colors required for the existence of an ${\bf m}$-distinguishable color code on $G_{\bf M}^c$.

\section{The Sunmao Construction}
\label{section:MT}
The complexity of constructing distinguishable color codes grows with the grid dimension, grid size and block size. One approach to reduce complexity is to decompose a grid with a given block size into sub-grids with smaller block sizes.  
Although this approach can be applied to grids of any dimension, we start our discussion with 1D grids first.  

Given $M, m \in \mathbb{Z}^+$ such that $m$ divides $M$, consider a
color mapping problem on a 1D grid $G^c_M$ with block size $m$. 
Let $I \geq 2$ be an integer.
We decompose $m$ into $I$ sub-block sizes so that
$m=m_1+\dots+m_I$, where $m_1,...,m_I \in \mathbb{Z}^+$.  
Let $d_1=0$ and for $1 < i \leq I$, define $d_{i+1}=d_i+m_i$.  
Note that $d_{I+1}=m$.  Define $M_i=m_iM/m \in \mathbb{Z}$ so that
$M=M_1+\dots+M_I$.

Decompose $G^c_M$  into $I$ sub-grids according to the sub-block sizes by first partitioning $G^c_M$ into $I$ mutually exclusive subsets, $S_1,\dots,S_I$, where $S_i$ consists of the following points:
\[
\{ x\in \mathcal{Z}_M: d_{i} \leq x {\rm ~mod~} m < d_{i+1} \}.
\]
For $1\leq i \leq I$, an element $x\in S_i$ can be expressed as $jm+d_{i}+x_r$, with  $0\leq j <M/m$ and $0\leq x_r <m_i$.
Note that $|S_i|=m_iM/m=M_i$.
The function, $\theta_i$, mapping $x$ to $jm_{i}+x_r$ defines a bijective map from $S_i$ to $\mathbb{Z}_{M_i}$. We can thus identify $S_i$ with grid $G^c_{M_i}$ via $\theta_i$.  
When we refer to $x$ in the text, it can be interpreted
either as an element in $G^c_M$, $jm+d_i+x_r$, or as an element in $G^c_{M_i}$, $jm_i+x_r$, depending on the context.

We refer to $\{G^c_{M_1},\ldots,G^c_{M_I}\}$ as the {\it sunmao decomposition} of $G^c_M$.
Conversely, given such a collection of grids, the {\it sunmao synthesis}, $G^c_{M_1}\oplus\cdots \oplus G^c_{M_I}$ refers to be the cyclic grid $G^c_M$ with partition $\{S_i\}$ such that $S_i=\theta_i^{-1}(G^c_{M_i})$.

$G^c_M$ can be partitioned into $M/m$ $m$-blocks tagged at grid points $am$ for some $a \in \mathbb{Z}_{M/m}$. We refer to such blocks as {\it aligned blocks}.  Each aligned $m$-block clearly contains one and only one aligned $m_i$-block of $G^c_{M_i}$ located at $am+d_i$ on $G^c_M$ and $am_i$ on $G^c_{M_i}$.  
This result is generalized to an arbitrary block in the following lemma.
The proof is straightforward and deferred to Appendix~\ref{appendix:sunmao 1D}.

\begin{lemma}\label{lemma:sunmao 1D}
Consider a 1D grid $G^c_M$ with block size $m$ such that
$m$ divides $M$.  Given a sunmao decomposition of $G^c_M$ into
sub-grids $\{G^c_{M_1},\ldots,G^c_{M_I}\}$, the following properties hold:
\begin{enumerate}

\item Each $m$-block of $G^c_M$  consists of one and only one $m_i$
block of $G^c_{M_i}$, $1\leq i \leq I$.
\item Let $B$ be the $m$-block tagged at $ x =jm+d_{i}+x_r\in \mathbb{Z}_M$, $0 \leq j < M/m$ and $0 \leq x_r  < m_i$. (In the following statements, the index $j$ ranges from 0 to $M/m-1$ and $j+1$ stands for $j+1 {\rm ~mod~} M/m$.)
\begin{enumerate}
\item If $x_r=0$ and $i=1$, then the decomposed block consists of one aligned $m_l$-block tagged at $jm_l\in G^c_{M_l}$ for $l, 1 \leq l \leq I$.
    \item If $x_r=0$ and $1<i\leq I$, then the decomposed block consists of one aligned $m_l$-block tagged at $jm_l\in G^c_{M_l}$ for $l$, $i \leq l \leq I$ and at $(j+1)m_l\in G^c_{M_l}$ for $l, 1 \leq l <i$.  
    \item If $0 < x_r  < m_i$, then the decomposed block consists of  one $m_i$-block tagged at $jm_i+x_r\in G^c_{M_i}$  and one aligned $m_l$-block at $jm_l\in G^c_{M_l}$ for $l, i<l  \leq I$ and one at $(j+1)m_l\in G^c_{M_l}$ for $l, 1 \leq l <i$. 
\end{enumerate}
\item If $ x,  x' \in \mathbb{Z}_M$ are distinct grid points, then the collection of decomposed blocks of $B(x)$ and $B(x')$ are distinct.
\end{enumerate}
\end{lemma}

Next, we show how to construct a color code for $G^c_M$ based on color codes on $G^c_{M_i}$.  Let $\Phi_i$ be a color mapping from $G^c_{M_i}$ to a set of distinct colors, $C_i$. The colors sets $C_i$ are required to be mutually exclusive.  

\begin{definition}
Let $C_M=\bigcup_{1 \leq i \leq I} C_i$.  Define a \emph{sunmao color mapping}, $\Phi$, from $G^c_M$ to $C_M$ by $\Phi(x)=\Phi_i(x)$ if $x\in G_{M_i}$.  We represent $\Phi$ by $(\Phi_1,\ldots,\Phi_I)$.
\end{definition}

\begin{proposition}\label{prop:sunmao code}
$\Phi$ is $m$-distinguishable if all $\Phi_i$'s are $m_i$-distinguishable, for $1\leq i \leq I$.
\end{proposition}

\begin{proof}
Suppose $B(x)$ and $B(x')$ have identical color multisets
and $x\ne x'$.
By Lemma \ref{lemma:sunmao 1D}, there is a sub-grid $G^c_{M_l}$ on
which $B(x)$ and $B(x')$ have different decomposed sub-blocks.  Under $\Phi_l$, the image sets must be different.  The color sets for different sub-grids are mutually exclusive, so
$B(x)$ and $B(x')$ cannot have identical color multisets, a contradiction,  
\end{proof}
The converse of Proposition \ref{prop:sunmao code} is not
valid. Indeed, the sunmao construction is useful exactly
because we can construct distinguishable color codes
based on non-distinguishable component codes.
Details are described in the following section.

It is possible to extend the sunmao construction to $n$-dimensional grids.  However, for the purpose of this paper, it suffices to focus on the following special case.

Given vectors ${\bf M}=(M_1,\ldots ,M_n)$ and ${\bf m}=(m_1, \ldots ,m_n)$ with $m_i,M_i,M'_i=M_i/m_i \in \mathbb{Z}^+$ for all $i$, consider a grid $G^c_{\bf M}=\mathbb{Z}_{M_1} \times \ldots \times \mathbb{Z}_{M_n}$ with block size $\bf m$.
Let $\{S_{\bf J}: {\bf 0}_n\leq  {\bf J}=(j_1,\ldots,j_n) < {\bf m}\}$ be a partition of $G^c_{\bf M}$ defined by:
$S_{\bf J}=\{{\bf J}+(l_1m_1,\ldots ,l_nm_n): 0\leq l_i <M'_i, \forall i, 1\leq i \leq  n\}$.
The mapping $\theta_{\bf J}:S_{\bf J}\to G^c_{\bf M'}=\mathbb{Z}_{M'_1}\times\ldots\times\mathbb{Z}_{M'_n}$ then can be defined by $\theta_{\bf J}\left((j_1+l_1m_1,\ldots,j_n+l_nm_n)\right)=(l_1,\ldots,l_n)$, 
where ${\bf M'}=(M'_1,\ldots,M'_n)$.  
Note that $\theta_{\bf J}$ is bijective, and thus $S_{\bf J}$ and $G^c_{\bf M'}$ are regarded as equivalent grids.
We use $G^c_{\bf M';J}$ to emphasize the cyclic grid that is associated with $S_{\bf J}$.

\begin{definition}
The collection of cyclic grids, $\{G^c_{\bf M';J}:\,{\bf 0}_n\leq {\bf J} < {\bf m}\}$ is defined to be the {\it unitary sunmao decomposition} of $G^c_{\bf M}$.
\end{definition}

Note that the total number of sub-grids in the above decomposition is $v({\bf m})$ and each of them has a grid size $\mathbb{Z}_{M'_1}\times\ldots\times\mathbb{Z}_{M'_n}$.  The following result is trivial to establish, but crucial to our proposed construction.

\begin{proposition}

In a unitary sunmao decomposition of an $n$-dimensional grid $G^c_{\bf M}$, each $\bf m$-block of $G^c_{\bf M}$ consists of one and only one element from each of the component sub-grids.
\end{proposition}

\section{Braid Code for 1D Grids}
\label{section:Braid 1}

To obtain an $m$-distinguishable color code via sunmao decomposition, we propose a construction
we refer to as {\it braid code}.
We will show that the number of color symbols required by
an important class of braid codes, the unitary braid codes, has the same asymptotic order as that of the optimal codes when the grid size tends to infinity.

\subsection{Construction Based on Repetitive Code}
Let $\ell$ be a positive integer that divides $M$. 
We regard grid $G^c_{M}$ as a concatenation of $M/{\ell}$ copies 
of $G^c_{\ell}$ in the subsequent discussion. 

\begin{definition}\label{definition:repetitive code}
Given a color code on $G^c_\ell$ defined by a color mapping, $\Gamma$, the \emph{repetitive code} on $G^c_M$, defined by the color mapping $\Phi$, is constructed by replicating $\Gamma$ to all of $G^c_{M}$, that is, $\Phi(x) \triangleq\Gamma(x \mod{\ell})$.
We refer to $G^c_\ell$  as a \emph{generator} of $G^c_M$ and $\Phi$ as being based on $\Gamma$.
\end{definition}

Given a sunmao decomposition of a 1D grid $G^c_M$ with sub-block sizes, $m_i$ and suppose there is an $m_i$-distinguishable color code on $G^c_{\ell_i}$ for some integer, $\ell_i$ for all $i$, it is possible to synthesize an $m$-distinguishable code for $G^c_M$ provided the system parameters satisfy certain conditions to be
detailed below.  The resulting
codes are referred to as {\it braid codes}.

Let $\lcm$ and $\gcd$ denote the least common multiple and the greatest common divisor, respectively.
The family of braid codes is defined as follows.

\begin{definition}\label{definition:braid code}
Consider a 1D grid $G^c_M$ with block size $m$, where $m$ divides $M$.  Let $\{G^c_{M_1},\ldots,G^c_{M_I}\}$ be a sunmao decomposition of $G^c_M$ so that $G^c_{M_i}$ has block size $m_i$.  
For $1 \leq i \leq I$, let $\Phi_i$ be a repetitive map on $G^c_{M_i}$ based on an $m_i$-distinguishable color code on a generator, $G^c_{\ell_i}$, with $\gcd(m_i,\ell_i)=c_i$, $\ell_i=gc_iq_i$, where $g$ is a positive integer satisfying $gc_i>m_i$ for all $i$.
Moreover, the grid size satisfies $M=mgQ$, where $Q=\lcm(q_1,\cdots,q_I)$.
A braid code on $G^c_M$ associated with  $\{G^c_{M_1},\ldots,G^c_{M_I}\}$ is defined by the color mapping, $\Phi_M=(\Phi_1, \ldots, \Phi_I)$.

\end{definition}

Note that under the stated conditions, $M_i$ is divisible by $l_i$ so the repetitive codes are well defined.  Note
also the condition $\gcd(m_i,\ell_i)=c_i$ is equivalent to $\gcd(m_i/c_i,gq_i)=1$.
Moreover, $g>1$ since $c_i$ is a factor of $m_i$.

There are two special classes of braid codes.
\begin{enumerate}
\item If we set $c_i=m_i$ for each $i$, then $\ell_i=gm_iq_i$ and $g>1$. The resulting braid code is referred to as \textit{class 1}.
\item If we set $c_i=1$ for each $i$, then $\ell_i=gq_i$ and $g>m_i$. The resulting braid code is referred to as \textit{class 2}.
\end{enumerate}

We note that the Synthetic Construction given in \cite{Paper1} is a special case of the braid code of class 1.
More precisely, it is the case that $I=2$ and $\gcd(q_1,q_2)=1$.

\begin{example}\label{ex:braid_code} \rm
Consider $M=75, m=5, I=2$ and $m_1=2, m_2=3$.
It follows that $M_1=30, M_2=45$ and the sub-grids are, respectively, $G^c_{30}$ and $G^c_{45}$.
As $M/m=15$, we may let $g=3,Q=5$ or $g=5, Q=3$.
Consider the following three braid codes.

(1) When $g=3,Q=5$, we pick $c_1=2,q_1=1$ and $c_2=3,q_2=5$, leading to $\ell_1=6$ and $\ell_2=45$.
The repetitive code, $\Phi_1$, on $G^c_{30}$ is based on a $2$-distinguishable color mapping, $\Gamma_1$, on $G^c_{6}$.
According to \cite{Paper1}, one needs a minimum of 3 colors for $\Gamma_1$.
Using the color set $\{a_1,a_2,a_3\}$, one candidate for $\Gamma_1$ is $a_1a_1a_2a_2a_3a_3$.
Replicating this pattern to all $G^c_{30}$ leads to the repetitive code $\Phi_1$:
\begin{equation}\label{eq:ex_braid_case1-1}
a_1a_1a_2a_2a_3a_3\ a_1a_1a_2a_2a_3a_3\ a_1a_1a_2a_2a_3a_3\ a_1a_1a_2a_2a_3a_3\ a_1a_1a_2a_2a_3a_3. 
\end{equation}
Here we use the convention that a color mapping on a grid is represented by a sequence of the images of the grid points.
Similarly, we need a $3$-distinguishable color mapping for $G^c_{45}$, denoted by $\Phi_2$.
Note that $\Gamma_2=\Phi_2$ in this case.
According to \cite{Paper1}, the minimum number of colors is $6$.
Using $\{b_1,b_2,b_3,b_4,b_5,b_6\}$ as a color set, $\Phi_2$ can be chosen by
\begin{equation}\label{eq:ex_braid_case1-2}
\begin{split}
&b_1b_1b_1b_2b_2b_2b_3b_3b_3\ b_1b_1b_6b_6b_3b_1b_5b_5b_2\ b_2b_4b_5b_3b_5b_3b_2b_4b_4 \\
&b_3b_3b_6b_2b_1b_4b_1b_4b_6\ b_2b_6b_2b_5b_1b_4b_3b_6b_5.
\end{split}
\end{equation}
Combining \eqref{eq:ex_braid_case1-1}--\eqref{eq:ex_braid_case1-2}, the resulting braid code on $G^c_{75}$ is:
\begin{align*}
&a_1a_1b_1b_1b_1\ a_2a_2b_2b_2b_2\ a_3a_3b_3b_3b_3\ a_1a_1b_1b_1b_6\ a_2a_2b_6b_3b_1 \\
&a_3a_3b_5b_5b_2\ a_1a_1b_2b_4b_5\ a_2a_2b_3b_5b_3\ a_3a_3b_2b_4b_4\ a_1a_1b_3b_3b_6 \\
&a_2a_2b_2b_1b_4\ a_3a_3b_1b_4b_6\ a_1a_1b_2b_6b_2\ a_2a_2b_5b_1b_4\ a_3a_3b_3b_6b_5.
\end{align*}

(2) When $g=3,Q=5$, if $c_1=1,q_1=1$ and $c_2=3,q_2=5$, then we get $\ell_1=3$ and $\ell_2=45$.
The repetitive code, $\Phi_1$, on $G^c_{30}$ is based on a 2-distinguishable color mapping, $\Gamma_1$, on $G^c_3$.
By~\cite{Paper1} again, 3 colors are needed.
Similar (1), with the color set $\{a_1,a_2,a_3\}$, the repetitive code $\Phi_1$ is then given by
\begin{equation}\label{eq:ex_braid_case2}
a_1a_2a_3a_1a_2a_3\ a_1a_2a_3a_1a_2a_3\ a_1a_2a_3a_1a_2a_3\ a_1a_2a_3a_1a_2a_3\ a_1a_2a_3a_1a_2a_3.
\end{equation}
By using the same $\Phi_2$ for $G^c_{45}$ as given in~\eqref{eq:ex_braid_case1-2}, the resulting braid code is
\begin{align*}
&a_1a_2b_1b_1b_1\ a_3a_1b_2b_2b_2\ a_2a_3b_3b_3b_3\ a_1a_2b_1b_1b_6\ a_3a_1b_6b_3b_1 \\
&a_2a_3b_5b_5b_2\ a_1a_2b_2b_4b_5\ a_3a_1b_3b_5b_3\ a_2a_3b_2b_4b_4\ a_1a_2b_3b_3b_6 \\
&a_3a_1b_2b_1b_4\ a_2a_3b_1b_4b_6\ a_1a_2b_2b_6b_2\ a_3a_1b_5b_1b_4\ a_2a_3b_3b_6b_5.
\end{align*}

(3) When $g=5,Q=3$, we can pick $c_1=1,q_1=3$ and $c_2=1,q_2=1$, leading to $\ell_1=15$ and $\ell_2=5$.
Using color set $\{a_1,a_2,a_3,a_4,a_5\}$ for $G^c_{15}$ and $\{b_1, b_2,b_3\}$ for $G^c_{5}$, the two based color mappings can set to be
\begin{align*}
&\Gamma_1: a_1a_1a_2a_2a_3a_3a_4a_4a_5a_5a_1a_3a_5a_2a_4 \\
&\Gamma_2: b_1b_1b_2b_2b_3.
\end{align*}
The two repetitive codes are given by
\begin{equation}\label{eq:ex_braid_case3-1}
\Phi_1:\ a_1a_1a_2a_2a_3a_3\ a_4a_4a_5a_5a_1a_3\ a_5a_2a_4a_1a_1a_2\ a_2a_3a_3a_4a_4a_5\ a_5a_1a_3a_5a_2a_4,
\end{equation}
and
\begin{equation}\label{eq:ex_braid_case3-2}
\begin{split}
\Phi_2:\ 
&b_1b_1b_2b_2b_3b_1b_1b_2b_2\ b_3b_1b_1b_2b_2b_3b_1b_1b_2\ b_2b_3b_1b_1b_2b_2b_3b_1b_1\\
&b_2b_2b_3b_1b_1b_2b_2b_3b_1\ b_1b_2b_2b_3b_1b_1b_2b_2b_3.
\end{split}
\end{equation}
Then, the resulting braid code on $G^c_{75}$ is:
\begin{align*}
&a_1a_1b_1b_1b_2\ a_2a_2b_2b_3b_1\ a_3a_3b_1b_2b_2\ a_4a_4b_3b_1b_1\ a_5a_5b_2b_2b_3 \\
&a_1a_3b_1b_1b_2\ a_5a_2b_2b_3b_1\ a_4a_1b_1b_2b_2\ a_1a_2b_3b_1b_1\ a_2a_3b_2b_2b_3 \\
&a_3a_4b_1b_1b_2\ a_4a_5b_2b_3b_1\ a_5a_1b_1b_2b_2\ a_3a_5b_3b_1b_1\ a_2a_4b_2b_2b_3.
\end{align*}
\end{example}

The three braid codes in Example~\ref{ex:braid_code} are all $5$-distinguishable on $G^c_{75}$.
The first and second braid codes use $9$ color symbols, while the third one uses $8$ only.
Note that the first one is class 1 and the third one is class 2.

\begin{lemma}\label{lemma:code distance}
Let $\Phi$ define a braid code.  Consider two $m_i$-blocks of $G^c_{M_i}$,
$B(x)$ and $B(x')$.
\begin{enumerate}
\item If $\Phi_i(B(x))=\Phi_i(B(x'))$, then the distance between $x$ and $x'$ over $G^c_{M_i}$ is divisible by $\ell_i=gc_iq_i$.
\item If both blocks are aligned on $G^c_{M_i}$, then $\Phi_i(B(x))=\Phi_i (B(x'))$ only if the distance between $x$ and $x'$ over $G^c_{M_i}$ is divisible by $gm_iq_i$.
\end{enumerate}
\end{lemma} 
\begin{proof}
Suppose $\Phi_i(B(x))=\Phi_i(B(x'))$.
Since $\Phi_i$ is a repetitive code based on an $m_i$-distinguishable color code on $G^c_{\ell_i}$, so the distance between $x$ and $x'$ over $G^c_{M_i}$ is divisible by $\ell_i$.
The first statement holds. 
Moreover, when $B(x)$ and $B(x')$ are aligned $m_i$-blocks, it is clear that the distance between $x$ and $x'$ over $G^c_{M_i}$ is also divisible by $m_i$.
By (i), the distance is divisible by $\lcm(\ell_i,m_i)$, which is equal to $gm_iq_i$ since $c_i$ divides $m_i$.  
\end{proof}

We now prove that the braid codes are $m$-distinguishable.

\begin{theorem}\label{theorem:construction}
The braid code given by Definition~\ref{definition:braid code} is $m$-distinguishable.
\end{theorem}
\begin{proof}
Let $\Phi$ define a braid code.
Suppose that the $m$-blocks at $x$ and at $x'$ have identical images, we claim that $x=x'$.

Since the grid and sub-grids are all cyclic, without lost of generality we can shift the indices of $x$ and $x'$ over $G^c_M$ so that $x'=0$ and $x=jm+d_i+x_r$ for some $0\leq j <J=M/m=gQ$, $1 \leq i \leq I$ and $0 \leq x_r < m_i$.
Note that all sub-blocks tagged at $x'$ are aligned and begin at position $0$ within their respective sub-grids.
Denote the $m$-block of $G^c_{M}$ at $x$ by $B$.
We further denote by $B_k$ the associated $m_k$-block of $G^c_{M_k}$, for $k=1,\ldots,I$.

We first consider the case where $0<x_r<m_i$, that is, $B_i$ is not aligned.
Observe that $B_i$ begins at the point $jm_i+x_r$ of $G^c_{M_i}$.
By Lemma~\ref{lemma:code distance}(i),
\begin{equation}\label{eq:braid-thm_1}
(jm_i+x_r)/(gc_i) \in \mathbb{Z}^+.
\end{equation}
If $i\geq 2$, $B_1$ exactly coincides with the aligned $m_1$-block beginning at $(j+1)m_1$ (instead of $jm_1$) of $G^c_{M_1}$.
According to Lemma~\ref{lemma:code distance}(ii), it
follows that $gm_1q_1$ divides $(j+1)m_1$, which implies $g$ divides $(j+1)$. Therefore, $gc_i$ is a factor of $(j+1)m_i$ and 
\begin{align*}
(j+1)m_i=jm_i+x_r+m_i-x_r
\end{align*}
is a multiple of $gc_i$.
By \eqref{eq:braid-thm_1}, it follows that $gc_i$ divides $(m_i-x_r)$.
However, $0<m_i-x_r<m_i$, which is strictly less than $gc_i$ by definition.  This is a contradiction.

If $i=1$, the aligned $m_2$-block $B_2$ begins at $jm_2$ of $G^c_{M_2}$.
By Lemma~\ref{lemma:code distance}(ii) again,  $gm_2q_2$ divides $jm_2$, implying $g$  divides $j$ and thus $gc_i$
divides $jm_i$.   It follows from \eqref{eq:braid-thm_1} that $gc_i$ is a factor of $x_r$, which is impossible due to $0<x_r<m_i<gc_i$.

It suffices to consider the case that $x_r=0$, that is, $B$ is an aligned $m$-block at $jm+d_i$ and all its sub-blocks are aligned sub-blocks of the corresponding sub-grids.

If $2\leq i\leq I$, the aligned block $B_i$ is at $jm_i$ and the aligned block $B_{i-1}$ is at $(j+1)m_{i-1}$.
Similar to previous arguments, we have $gq_i$ divides $j$ and $gq_{i-1}$ divides $(j+1)$, 
which is impossible since $g>1$ is the common factor of $gq_i$ and $gq_{i-1}$.

If $i=1$, that is, $x=jm$, then each aligned sub-block $B_k$ is at $jm_k$ for $k=1,\ldots,I$.
In this case, by Lemma~\ref{lemma:code distance}(ii), $gm_kq_k$ divides $jm_k$, which implies $gq_k$ is
a factor of $j$.  Therefore, $j$ divides $\lcm(gq_1,\ldots,gq_I)=gQ$, which implies $j=0$ and $x=x'$.
This completes the proof.
\end{proof}

For a sunmao decomposition into unitary block sizes, so that $m_i=1$ for all $i$, there is only one class of braid code. 
We call such a braid code a {\it unitary braid code}.
In this case, $\ell_i=gq_i$ for all $1\leq i \leq I$, $g\geq 2$, and $M=gmQ$, where $Q=\lcm(q_1, \ldots, q_I )$.

We restate here some results for unitary braid code from Lemma \ref{lemma:sunmao 1D} and Lemma~\ref{lemma:code distance} for ease of reference.

\begin{corollary}\label{corol:braid code 2}
If  $\Phi$  defines a unitary braid code on $G^c_M$ then the image set under $\Phi$ of any $m$-block of $G^c_M$ contains only elements with multiplicity 1.
\end{corollary}

\begin{corollary}\label{corol:braid code 1}
Let $\Phi$ define a unitary braid code on $G^c_{M_1}\times\dots\times G^c_{M_I}$ and for $i, 1\leq i \leq I$, let $G^c_{\ell_i}$ be the generator of $G^c_{M_i}$, then $x, x'\in G^c_{M_i}$ maps to the same color symbol if and only if their distance, $|x-x'|$, in $G^c_{M_i}$ is divisible by $\ell_i$. 
\end{corollary}

Corollary \ref{corol:braid code 2} guarantees that all codewords of a unitary braid code do not contain repeated colors.  Hence, they are more closely related to universal cycles \cite{HJZ09} rather than de Bruijn sequences.

\subsection{Performance Comparison of Braid Codes}

One can see from Example~\ref{ex:braid_code} that the structure of the braid codes may vary according to
the choice of parameters $g,Q,c_i$ and $q_i$.
As they are all $m$-distinguishable by Theorem~\ref{theorem:construction}, it is of interest to investigate which choice of parameters results in the braid code that uses the least color symbols.

Fix the grid size $M$ and each sub-block size $m_i$, $1\leq i\leq I$ and $I\geq 2$.
By the construction of braid codes, the number of color symbols required is
\begin{equation}\label{eq:braid-code-color-number}
K^c_{\ell_1}(m_1)+K^c_{\ell_2}(m_2)+\cdots+K^c_{\ell_I}(m_I),
\end{equation}
where each $\ell_i=gc_iq_i$ for some $g,c_i,q_i$ satisfying $M/m=g\lcm(q_1,\ldots,q_I)$ and $gc_i>m_i$ for all $i$.
We assume without loss of generality that $m_1\leq m_2\leq \cdots\leq m_I$.
As the value in~\eqref{eq:braid-code-color-number} is invariant under permutations of $\ell_{i},\ldots,\ell_{j}$ when $m_i=\cdots=m_j$, $i<j$, we fix the order $\ell_{i}\leq\cdots\leq\ell_j$ and call it the \emph{canonical} ordering.
A canonical list of generators $\boldsymbol{\ell}=(\ell_1,\ldots,\ell_I)$ is called \emph{optimal} if it minimizes~\eqref{eq:braid-code-color-number}.
Note that an optimal list of generators may not be unique, but the corresponding minimum value is unique.

We will show that, in many cases, the class 2 braid code requires the smallest number of color symbols when the sub-block size does not exceed $3$.
Before that, we need the following lemma, whose proof is deferred to Appendix~\ref{appendix:braid-code-color-number}.

\begin{proposition}\label{prop:braid-code-color-number}
For any optimal canonical list of generators $\boldsymbol{\ell}=(\ell_1,\ldots,\ell_I)$, if $m_I\leq 3$, then $\ell_1\leq \ell_2\leq \cdots\leq\ell_I$.
\end{proposition}

For simplicity, we consider the special case in which the smallest prime factor of $M/m$ is larger than $m_I$.

\begin{theorem}\label{thm:braid-code-color-number}
Let $\boldsymbol{\ell}=(\ell_1,\ldots,\ell_I)$ be an optimal canonical list of generators.
If $m_I\leq 3$ and the smallest prime factor of $M/m$ is larger than $m_I$, then $\boldsymbol{\ell}$ corresponds to a class 2 braid code.
\end{theorem}
\begin{proof}
Recall that $\ell_i$ is in the form $gc_iq_i$ for all $i$, where $\lcm(\ell_1,\ldots,\ell_I)=M/m=gQ$.
It suffices to show that $c_1=\cdots=c_I=1$.
As the smallest prime factor of $M/m$ is larger than $m_I$, we have $g>m_I$.

Suppose to the contrary that $t$ is the smallest index such that $c_t>1$.
If $t=1$, then $(gq_1,\ell_2,\ldots,\ell_I)$ is a list of generators which requires fewer color symbols.
If $t>1$, by Proposition~\ref{prop:braid-code-color-number}, we have $gq_{t-1}\leq gc_tq_t$. ($c_{t-1}=1$ in this sub-case.)
Let $s$ be the smallest index such that $q_s\geq q_t$.
Note that such an $s$ is well-defined.
Now, let 
\begin{align*}
    \ell'_i = \begin{cases}
        gq_t, & \text{if }i=s, \\
        gq_{i-1}, & \text{if }s+1\leq i\leq t, \\
        \ell_i, & \text{otherwise}.
    \end{cases}
\end{align*}
One can check that $\lcm(\ell'_1,\ldots,\ell'_I)=\lcm(\ell_1,\ldots,\ell_I)$.
It is easy to see that $\ell'_i\leq \ell_i$ for all $i$, and thus $\boldsymbol{\ell'}=(\ell'_1,\ldots,\ell'_I)$ requires fewer color symbols.
\end{proof}

\begin{remark}\rm
The proof of Proposition~\ref{prop:braid-code-color-number} relies on the exact values of $K_m(\ell)$, which have been determined in~\cite{Paper1} for $1\leq m\leq 3$.
However, since~\cite{Paper1} also provides the asymptotic behavior $K_m(\ell)=\Theta(\ell^{1/m})$, Proposition~\ref{prop:braid-code-color-number} holds asymptotically for any $m$.
In this context, Theorem~\ref{thm:braid-code-color-number} also remains valid asymptotically for all $m$, thus extending the original result beyond the restriction $m_I\leq 3$.
\end{remark}

\subsection{Construction for a General Size Unitary Braid Code}
\label{section: 1D generalized braid}

The previous construction method for braid code is defined only for grids of size $M=gmQ$, where $g\geq 2$ and $Q$ is the $lcm$ of a collection of $m$ positive integers.   If $M_r$ is strictly less than $M$, a flat grid $G_{M_r}$ can be viewed as a sub-grid of $G_{M}$.  It is easy to see that the restriction of a braid code on $G_{M}$ to $G_{M_r}$ defines an $m$-distinguishable code in the coding area of the sub-grid. The situation for a cyclic grid is more complicated. 

In Example \ref{ex:braid_code} three braid codes on $G^c_{75}$ are presented.  If we restrict the first code to $G^c_{19}$, the resulting code is not 5-distinguishable. Note that in this case $M_r=19$ is not a multiple of $m=5$.  If we restrict the  third code to $G^c_{30}$, the resulting code is also not 5-distinguishable and in this case $m$ divides $M_r$.  For unitary braid code, the following result holds if $m$ does not divide
$M_r$.

\begin{proposition}\label{prop:restriction}
The restriction of a color mapping of a unitary braid code on $G^c_M$ to a sub-grid of size $M_r$ defines an $m$-distinguishable code if $m$ does not divide $M_r$.
\end{proposition}

The proof of this proposition is simply based on tedious case by case checking and is relegated to Appendix \ref{appendix:restriction}.

For $M_r/m=J \in \mathbb{Z}$ and $J \geq 2$., we propose a simple modification of a unitary braid code on $G^c_M$ to obtain an $m$-distinguishable on $G^c_{M_r}$.  (Note that for $J=1$ there is no $m$-distinguishable code for $G^c_{M_r}$.)
Let $C_i$ denote the color set for sub-grid $G^c_{M_i}$ for $1\leq  i \leq I$. 

Since the image sets of $B_m(0)$ and $B_m((J-1)m)$ are different, there is a grid point, $j$, such that
$\Phi(j) \ne \Phi((J-1)m+j)$.   By cyclically shifting the grid indices, we can assume without lost of generality that $\Phi(0) \ne \Phi((J-1)m)$.
Define $\Phi'$ to be the restriction of $\Phi$ to $G^c_{M_r}$ except that for $1\leq i <m$, $\Phi'((J-1)m+i)$ is mapped to $c^*\triangleq \Phi((J-1))m$, which by choice is distinct from $\Phi(0)$. 

\begin{proposition}\label{prop:general size}
$\Phi'$, obtained from modifying a unitary braid code, defines an $m$-distinguishable color code on $G^c_{M_r}$.
\end{proposition}

\begin{proof}
Only $m$-blocks containing grid points from $(J-1)m+1$ to $Jm-1$  have image sets that are different from the original braid code. We only need to show that the image sets of $m$- blocks tagged at grid points
ranging from $(J-2)m+2$ to $Jm-1$ are different from each other and are different from those blocks that do not contain these listed points. The latter claim is easy to see since the two types of image sets contain different number of elements from $C_1$, counting multiplicity. To prove the first claim, first note that $C_1 \cap \Phi'(B_m(J-2)m+i)$ contains $i$ copies of $c^*$  for $2\leq i \leq m$ and no other elements from $C_1$. For $m<i \leq 2m-1$, $C_1 \cap \Phi'(B_m(J-2)m+i)$ contains $2m-i$ copies of $c^*$ and an additional element of $C_1$ different from $c^*$.  So these images sets are all distinct and $\Phi'$ defines an $m$-distinguishable code. 
\end{proof}

One can check that for a general braid code, if we replace $c^*$ with a color that is distinct from all elements in all $C_i$'s, the resulting color code is also $m$-distinguishable.  

We continue to refer to the restricted or modified codes
as braid codes; the original constructions
satisfying specific parameter conditions will be 
referred to as {\it standard codes}.

\section{Braid Code for General Dimension}
\label{section:Braid n}

Let $G^c_{\bf M}$ an $n$-dimensional grid.  We can extend the definition of braid code to $G^c_{\bf M}$ by using {\it product color code}, which was first introduced in \cite{Paper1}.  We present the basic idea of a product color code here in a slightly more general context.

First, we introduce some relevant 
mathematical notation.
Let ${\bf k}=(k_1,\ldots,k_n)$ and $\mathcal{C}_n$ be a set of  $v({\bf k})$ color symbols.  
We represent $\mathcal{C}_n$ as $[v({\bf k}) ]$ and index its elements by $n$-tuples, $(c_1,\ldots,c_n)$.   
Hence, we can regard $\mathcal{C}_n$ as $[k_1]\times\ldots\times [k_n]$.

\subsection{Product Color Code}
\label{section:product-code}

Let $\Phi^{(i)}\in\mathcal{C}_{G^c_i,k_i }$ define a color code from $G^c_i$ to $[k_i]$.  For the product grid $G^c \triangleq G^c_1 \times \ldots \times G^c_n$ define a {\it product color code} (or a product code, for short) by a color mapping, 
$\Phi \in\mathcal{C}_{G^c,|\mathcal{C}_n| }$, so that:
\[
\Phi({\bf x})=(\Phi^{(1)}(x_1),\ldots,\Phi^{(n)}(x_n)).
\]

Let ${\bf m}=(\mi{1},\ldots,\mi{n})$.  
An $\bf m$-block, $B$, of $G^c_{\bf M}$ can be represented in product form as $B_1\times\ldots\times B_n$, in which $B_i$ is an $m_i$-block of $G^c_{M_i}$.
 

The following result is due to~\cite[Proposition 1]{Paper1}.

\begin{proposition}[\!\!\!\cite{Paper1}]\label{prop:product-code}
The product code, $\Phi$, on $G^c_{\bf M}$ is ${\bf m}$-distinguishable if and only if
$\Phi^{(i)} $ is $\mi{i}$-distinguishable for all $i, 1\leq i \leq n$.
\end{proposition}

For grids with unitary block size, ${\bf 1}_n$, product code plays a special role as highlighted by the following result.

\begin{proposition}\label{prop:unitary product code}
Let $G^c_{\bf M}$ be an $n$-dimensional grid with unitary block size, ${\bf 1}_n$.  if $\Phi$ defines a ${\bf 1}_n$-distinguishable color code, the code can be represented as a product code.
\end{proposition}
\begin{proof}
The code defined by $\Phi$ is ${\bf 1}_n$-distinguishable if and only if $\Phi$ maps every grid point to a unique color.  So 
the minimal image set contains $\nu({\bf M})$ colors.  We can represent such an image set as a product color set and identify color $\Phi(\bf x)$  with $\bf x$.  In other words
\[
\Phi({\bf x})=(x_1,\ldots,x_n)  \in \mathcal{C}_n.
\]
\end{proof}

\subsection{$n$-dimensional Unitary Sunmao Code}
\label{section:repetitive code}

Although simple, product code does not
appear to be efficient.  Yet, they play an important
role in unitary sunmao decomposition for multi-dimension grids.

Consider an $n$-dimensional grid  $G^c_{\bf M}=\mathbb{Z}_{M_1} \times \ldots \times \mathbb{Z}_{M_n}$ with block size defined by ${\bf m}=(\mi{1}, \ldots ,\mi{n}), \mi{i} \in \mathbb{Z}^+$, where $M'_i=M_i/\mi{i}$ is an integer strictly greater than 1.  
Let ${\bf M'}=(M'_1,\ldots,M'_n)$.
Let $\mathcal{U}=\{G^c_{\bf M';J}: {\bf 0}_n \leq {\bf J}=(J_1,\ldots,J_n) < {\bf m}\}$ be a unitary sunmao decomposition of $G^c_{\bf M}$.
In other words, $B_{\bf m}(\bf x)$ contains one and only one element from each sub-grid $G^c_{\bf M';J}\cong S_{\bf J}$, which is located at ${\bf x+J}$. 
There are a total of $v({\bf m})$ sub-grids in the decomposition.  
All of which are of the same size and equivalent to $\mathbb{Z}_{M'_1} \times \ldots \times \mathbb{Z}_{M'_n}$. 
Let $\Phi_{\bf J}:G^c_{\bf M';J} \rightarrow \mathcal{C}_{\bf J}$ be a color mapping that defines a product code on $G^c_{\bf M';J}$ with $|\mathcal{C}_{\bf J}|=k_{\bf J}=k^{(1)}_{\bf J}\ldots k^{(n)}_{\bf J}$ for some $k^{(i)}_{\bf J}$'s.  
Note that sets $\mathcal{C}^{(i)}_{\bf J}$ and $\mathcal{C}^{(i')}_{{\bf J}'}$  are required to be disjoint if ${\bf J}\ne{\bf J}'$ or $i \ne i'$. 

We define an $n$-dimensional {\it unitary sunmao code}, or unitary code, for short, as follows.  Let ${\bf x}=(x_1,\ldots,x_n)$ be a grid point in $G^c_{\bf M}$.  We can represent $x_i$ as $j_i\mi{i}+r_i$, with $0\leq j_i <M'_i$ 
and $0\leq r_i <\mi{i}$.  Let ${\bf j}({\bf x})=(j_1,\ldots,j_n)$ and ${\bf r}({\bf x})=(r_1,\ldots,r_n), $ then the color mapping is defined as
\[
\Phi({\bf x}) \triangleq\Phi_{\bf r (x)}({\bf j}(x))\triangleq(\Phi^{(1)}_{\bf r(x)}(j_1),\ldots,\Phi^{(n)}_{\bf r(x)}(j_n)).
\]

The color set of $\Phi$, $\mathcal{C}$, is equal to $\bigcup_{\bf J} \mathcal{C}_{\bf J}$ with cardinality $\sum_{\bf J}k_{\bf J}$.  
For each $\bf J$, ${\bf 0}_n\leq {\bf }J < {\bf m}$, $\mathcal{C}_{\bf J}$ is an $n$-dimensional product set, we define an operator, $\mathcal{P}^{(i)}_{\bf J}$ from $\mathcal{C}_{\bf J}$ to $\mathcal{C}^{(i)}_{\bf J}$ by coordinate projection so that $\mathcal{P}^{(i)}_{\bf J}(\Phi_{\bf J}({\bf x}))=\Phi^{(i)}_{\bf J}(x_i)$.
Let $\mathcal{C}_i=\bigcup_{\bf J}\mathcal{C}^{(i)}_{\bf J}$. 
We then define a projection operator, $\mathcal{P}_i:\mathcal{C}\rightarrow\mathcal{C}_i$ by
\[
\mathcal{P}_i(\Phi({\bf x}))=
\Phi^{(i)}_{\bf r(x)}({j_i}).
\]
Given an $m$-block, $B$, define $\mathcal{P}_i(\Phi(B))$ as the collection of projected images set of $B$ onto the $i$-th coordinate
of the color product set.
It is straightforward to check that for ${\bf x}=(j_im_i+r_i;{\bf x}^-_i)$,
\begin{equation}\label{eq:nD braid code}
\begin{split}
\mathcal{P}_i&(\Phi(B({\bf x}))) = \\
&\left(
\bigcup_{r_i \leq s < m_i} \Omega^{(i)}_s (j_i) \right)
\bigcup \left(
\bigcup_{0 \leq s < r_i}\Omega^{(i)}_s (j_i+1) \right),
\end{split}
\end{equation}
where $\Omega^{(i)}_r = \bigcup_{{\bf 0}_{n-1} \leq {\bf J}^-_i \leq {\bf m}^-_i}
\Phi^{(i)}_{(r;{\bf J}^-_i)}$.

The following result is a crucial property of a unitary code.
The proof is very similar in idea to that of Proposition~\ref{prop:product-code} and is provided in Appendix~\ref{appendix:unitary code} for the sake of completeness.

\begin{proposition}\label{prop:unitary code}
An $n$-dimensional unitary code defined by $\Phi$ is $\bf m$-distinguishable if and only if for any $B({\bf x})$, an $m$-block tagged at ${\bf x}=(x_1,\ldots,x_n)\in G^c_{\bf M}$,  and any $l,1\leq l \leq n$, $\mathcal{P}_l(\Phi(B({\bf x})))$ uniquely determines $x_l$.
\end{proposition}

Although $\mathcal{P}_i(\Phi(B({\bf x})))$ can uniquely determine $x_i$ in the 1D sub-grid $\mathbb{Z}_{M_i}$, the multiset does not constitute a valid color codeword, since its cardinality is $v({\bf m})$, which is different from $m_i$ unless $m_j=1$ for $j\ne i$.   Nevertheless, it is possible to construct a family of unitary sunmao codes for $\mathbb{Z}_{M_i}$ from $\Phi$.

\subsection{$n$-Dimensional Repetitive Code}
Consider a general $n$-dimensional grid  $G^c_{\bf M}=\mathbb{Z}_{M_1} \times \ldots \times \mathbb{Z}_{M_n}$.
We can extend the idea of a repetitive code to an $n$-dimensional grid by using a product code as a generator code. 

Let $\ell_i$ be a positive factor of $M_i$ for $1\leq i \leq n$, the grid $G^c_{\bf P}=\mathbb{Z}_{\ell_1}\times\ldots\times\mathbb{Z}_{\ell_n}$ is referred to as a generator of $G^c_{\bf M}$.  

\begin{definition}\label{definition:n-dim repetitive code}
Given a product code of $G^c_{\bf P}$ defined by the mapping $\Gamma=(\Gamma^{(1)},\ldots,\Gamma^{(n)})$ , the repetitive code on $G^c_{\bf M}$ based on $\Gamma$ is defined by the mapping $\Phi({\bf x})=(\Phi^{(1)}(x_1),\ldots,\Phi^{(n)}(x_n)) \triangleq (\Gamma^{(1)}(x'_1),\ldots,\Gamma^{(n)}(x'_n))$
where $x'_i=x_i {\rm ~mod~} \ell_i$, $0\leq x'_i <\ell_i$,
\end{definition}

\begin{lemma}\label{lemma:n-dim repetitive code}
Suppose $G^c_{\bf M}$ has a block size defined by the unit vector ${\bf 1}_n$ and the color code defined by $\Gamma$ is ${\bf 1}_n$-distinguishable on $G^c_{\bf P}$. Then two grid points, ${\bf x}=(x_1,\ldots,x_n)$ and ${\bf y}=(y_1,\ldots,y_n)$ are mapped to the same image by $\Gamma$ 
if and only if $|x_i-y_i|$ is divisible by $\ell_i$ 
for all $i, 1\leq i \leq n$.
\end{lemma}
\begin{proof}
Grid points $\bf x$ and $\bf y$ map to the same image if and only if for all $i, 1\leq i \leq n$, $\Gamma^{(i)}(x'_i)=\Gamma^{(i)}(y'_i)$ where $x'_i=x_i {\rm ~mod~} \ell_i$, $y'_i=y_i {\rm ~mod~} \ell_i$, $0\leq x'_i,y'_i<\ell_i$.   In other words, $\Gamma(x'_1,\ldots,x'_n)=\Gamma(y'_1,\ldots,y'_n)$.  Since $\Gamma$ is ${\bf 1}_n$-distinguishable on $G^c_{\bf P}$, this implies $x'_i=y'_i$ for all $i$.  
Due to the repetitive code structure, this implies 
$|x_i-y_i|$ is divisible by $\ell_i$.
\end{proof}

\subsection{$n$-dimensional Unitary Braid Code}
\label{section:unitary braid code}
An $n$-dimensional {\it unitary braid code} is a special code based on a unitary sunmao decomposition.

For each coordinate  $i$, let
$\mathcal{Q}_i \triangleq\{q^{(i)}_{\bf J}:{\bf 0}_n \leq {\bf J < m} \}$ be a multiset of strictly positive integers.    We denote the least common multiple of all elements in $\mathcal{Q}_i$ by $Q_i$.  
Note that the number of elements in $\mathcal{Q}_i$ is equal to $v({\bf m})$, the volume of an $\bf m$-block.  
We do not impose extra conditions on the
choice of $q^{(i)}_{\bf J}$ for different $i$ and $\bf J$.

Given an integer $g \geq 2$ and $M_i=gm_iQ_i$, define a grid size vector ${\bf M}=(M_1,\ldots ,M_n)$.  
Consider an $n$-dimensional grid $G^c_{\bf M}=\mathbb{Z}_{M_1} \times \ldots \times \mathbb{Z}_{M_n}$ with block size $\bf m$.
Let $M'_i =M_i /m_i$.
Note that $\{G^c_{\bf M';J}\}$ is unitary sunmao decomposition.
For each $G^c_{\bf M';J}$ we associate with it a generator, 
\[
G^c_{\bf M';J}=\mathbb{Z}_{\ell^{(1)}_{\bf J}}\times\ldots\times\mathbb{Z}_{\ell^{(n)}_{\bf J}},
\]
where $\ell^{(i)}_{\bf J}=gq^{(i)}_{\bf J}$, $q^{(i)}_{\bf J}\in \mathcal{Q}_i$ for all $i$. 
Note that $M'_i/\ell^{(i)}_{\bf J}=Q_i/q^{(i)}_{\bf J}$, which implies that $\ell^{(i)}_{\bf J}$ divides $M'_i$ for all $i$. 

Let $\Gamma_{\bf J}$ define an ${\bf 1}_n$-distinguishable product code on $G^c_{\bf M';J}$.  It follows that each grid point in the generator is mapped to a distinct color, so 
\[
g^{n} {\sum_{{\bf 0}_n \leq {\bf J < m}}\prod_{1\leq i\leq n} q^{(i)}_{\bf J}} 
\]
distinct colors are required.

\begin{definition}\label{definition:unitary braid code}
Let $G^c_{\bf M}$ be a $n$-dimensional grid with a block size $\bf m$ and $\mathcal{U}=\{G^c_{\bf M';J}: {\bf 0}_n \leq {\bf J <M}\}$ be a unitary sunmao decomposition with a color mapping $\Phi_{\bf M}=\{\Phi_{\bf J}\}$.   
The color code defined by $\Phi_{\bf M}$  is a unitary braid code if each $\Phi_{\bf J}$ is a repetitive code based on a ${\bf 1}_n$-distinguishable code on the generator defined by $\Gamma_{\bf J}$. 
\end{definition}

We provide a 2D unitary braid code example for illustration.
Let ${\bf m}=(2,2), g=2$, $q^{(1)}_{0,0}=1$, $q^{(1)}_{0,1}=1$, $q^{(1)}_{1,0}=2$, $q^{(1)}_{1,1}=3$, $q^{(2)}_{0,0}=3$, $q^{(2)}_{0,1}=2$, $q^{(2)}_{1,0}=1$, $q^{(2)}_{1,1}=1$.  
Hence, $Q_1=Q_2=6$.
Consider a grid $G^c_{\bf M}=\mathbb{Z}_{24}\times\mathbb{Z}_{24}$ with a block size $\bf m$.  
The unitary sunmao decomposition consists of 4 sub-grids, each of which is equivalent to $\mathbb{Z}_{12}\times\mathbb{Z}_{12}$, i.e., ${\bf M'}=(12,12)$.
The generator for $G^c_{{\bf M'};0,0}$ is $\mathbb{Z}_{2q^{(1)}_{0,0}}\times\mathbb{Z}_{2q^{(2)}_{0,0}}=\mathbb{Z}_2\times\mathbb{Z}_6$.
The generator for $G^c_{{\bf M'};0,1}$ is $\mathbb{Z}_{2q^{(1)}_{0,1}}\times\mathbb{Z}_{2q^{(2)}_{0,1}}=\mathbb{Z}_2\times\mathbb{Z}_4$.
The generator for $G^c_{{\bf M'};1,0}$ is $\mathbb{Z}_{2q^{(1)}_{1,0}}\times\mathbb{Z}_{2q^{(2)}_{1,0}}=\mathbb{Z}_4\times\mathbb{Z}_2$.
The generator for $G^c_{{\bf M'};1,1}$ is $\mathbb{Z}_{2q^{(1)}_{1,1}}\times\mathbb{Z}_{2q^{(2)}_{1,1}}=\mathbb{Z}_6\times\mathbb{Z}_2$.

The resulting color map $\Phi$ is shown in Fig.~\ref{fig:Fig1}.
\begin{figure}[t]
    \centering
    \hspace{-1.5cm}

{\footnotesize
$\Aa\Ba\Ab\Bb \Aa\Ba\Ab\Bb \Aa\Ba\Ab\Bb \Aa\Ba\Ab\Bb \Aa\Ba\Ab\Bb \Aa\Ba\Ab\Bb$ \\
$\Ca\Da\Cb\Db \Cc\Dc\Cd\Dd \Ca\De\Cb\Df \Cc\Da\Cd\Db \Ca\Dc\Cb\Dd \Cc\De\Cd\Df$ \\
$\Ac\Bc\Ad\Bd \Ac\Bc\Ad\Bd \Ac\Bc\Ad\Bd \Ac\Bc\Ad\Bd \Ac\Bc\Ad\Bd \Ac\Bc\Ad\Bd$ \\
$\Ce\Dg\Cf\Dh \Cg\Di\Ch\Dj \Ce\Dk\Cf\Dl \Cg\Dg\Ch\Dh \Ce\Di\Cf\Dj \Cg\Dk\Ch\Dl$ \\
$\Ae\Be\Af\Bf \Ae\Be\Af\Bf \Ae\Be\Af\Bf \Ae\Be\Af\Bf \Ae\Be\Af\Bf \Ae\Be\Af\Bf$ \\
$\Ca\Da\Cb\Db \Cc\Dc\Cd\Dd \Ca\De\Cb\Df \Cc\Da\Cd\Db \Ca\Dc\Cb\Dd \Cc\De\Cd\Df$ \\
$\Ag\Bg\Ah\Bh \Ag\Bg\Ah\Bh \Ag\Bg\Ah\Bh \Ag\Bh\Ah\Bh \Ag\Bh\Ah\Bh \Ag\Bh\Ah\Bh$ \\
$\Ce\Dg\Cf\Dh \Cg\Di\Ch\Dj \Ce\Dk\Cf\Dl \Cg\Dg\Ch\Dh \Ce\Di\Cf\Dj \Cg\Dk\Ch\Dl$ \\
$\Ai\Ba\Aj\Bb \Ai\Ba\Aj\Bb \Ai\Ba\Aj\Bb \Ai\Ba\Aj\Bb \Ai\Ba\Aj\Bb \Ai\Ba\Aj\Bb$ \\
$\Ca\Da\Cb\Db \Cc\Dc\Cd\Dd \Ca\De\Cb\Df \Cc\Da\Cd\Db \Ca\Dc\Cb\Dd \Cc\De\Cd\Df$ \\
$\Ak\Bc\Al\Bd \Ak\Bc\Al\Bd \Ak\Bc\Al\Bd \Ak\Bc\Al\Bd \Ak\Bc\Al\Bd \Ak\Bc\Al\Bd$ \\
$\Ce\Dg\Cf\Dh \Cg\Di\Ch\Dj \Ce\Dk\Cf\Dl \Cg\Dg\Ch\Dh \Ce\Di\Cf\Dj \Cg\Dk\Ch\Dl$ \\
$\Aa\Be\Ab\Bf \Aa\Be\Ab\Bf \Aa\Be\Ab\Bf \Aa\Be\Ab\Bf \Aa\Be\Ab\Bf \Aa\Be\Ab\Bf$ \\
$\Ca\Da\Cb\Db \Cc\Dc\Cd\Dd \Ca\De\Cb\Df \Cc\Da\Cd\Db \Ca\Dc\Cb\Dd \Cc\De\Cd\Df$ \\
$\Ac\Bg\Ad\Bh \Ac\Bg\Ad\Bh \Ac\Bg\Ad\Bh \Ac\Bg\Ad\Bh \Ac\Bg\Ad\Bh \Ac\Bg\Ad\Bh$ \\
$\Ce\Dg\Cf\Dh \Cg\Di\Ch\Dj \Ce\Dk\Cf\Dl \Cg\Dg\Ch\Dh \Ce\Di\Cf\Dj \Cg\Dk\Ch\Dl$ \\
$\Ae\Ba\Af\Bb \Ae\Ba\Af\Bb \Ae\Ba\Af\Bb \Ae\Ba\Af\Bb \Ae\Ba\Af\Bb \Ae\Ba\Af\Bb$ \\
$\Ca\Da\Cb\Db \Cc\Dc\Cd\Dd \Ca\De\Cb\Df \Cc\Da\Cd\Db \Ca\Dc\Cb\Dd \Cc\De\Cd\Df$ \\
$\Ag\Bc\Ah\Bd \Ag\Bc\Ah\Bd \Ag\Bc\Ah\Bd \Ag\Bc\Ah\Bd \Ag\Bc\Ah\Bd \Ag\Bc\Ah\Bd$ \\
$\Ce\Dg\Cf\Dh \Cg\Di\Ch\Dj \Ce\Dk\Cf\Dl \Cg\Dg\Ch\Dh \Ce\Di\Cf\Dj \Cg\Dk\Ch\Dl$ \\
$\Ai\Be\Aj\Bf \Ai\Be\Aj\Bf \Ai\Be\Aj\Bf \Ai\Be\Aj\Bf \Ai\Be\Aj\Bf \Ai\Be\Aj\Bf$ \\
$\Ca\Da\Cb\Db \Cc\Dc\Cd\Dd \Ca\De\Cb\Df \Cc\Da\Cd\Db \Ca\Dc\Cb\Dd \Cc\De\Cd\Df$ \\
$\Ak\Bg\Al\Bh \Ak\Bg\Al\Bh \Ak\Bg\Al\Bh \Ak\Bg\Al\Bh \Ak\Bg\Al\Bh \Ak\Bg\Al\Bh$ \\
$\Ce\Dg\Cf\Dh \Cg\Di\Ch\Dj \Ce\Dk\Cf\Dl \Cg\Dg\Ch\Dh \Ce\Di\Cf\Dj \Cg\Dk\Ch\Dl$ \\
}
   \caption{A 2D unitary braid code example.  The color set for $\Phi_{0,0}$ is $\{\Aa,\ldots,\Aj,\Ak,\Al\}$, for $\Phi_{0,1}$, 
   $\{\Ba,\ldots,\Bh\}$, for
$\Phi_{1,0}$, $\{ \Ca,\ldots,\Ch \}$, and for $\Phi_{1,1}$,
$\{\Da,\ldots,\Dj,\Dk,\Dl\}$.}  \label{fig:Fig1} 
     \vspace{-0.2cm}
\end{figure} 
It is worth noting that the color symbol sets have a hidden structure due to the product code construction. 
Take, for example, the color symbol set for $\Phi_{0,1}$. The 8 color symbols have an ordered-pair representation so that 
$b_0=(0,0),b_1=(1,0)$,  
$b_2=(0,1),b_3=(1,1)$, 
$b_4=(0,2),b_5=(1,2)$, 
$b_6=(0,3),b_7=(1,3)$.
If $b_0$ and $b_1$ are swapped, then $(b_2, b_3)$,
$(b_4,b_5)$ and $(b_6,b_7)$ should also be swapped
to preserve the product code structure.

\begin{theorem}\label{theorem:unitary braid code}
An $n$-dimensional unitary braid code given by Definition~\ref{definition:unitary braid code} is $\bf m$-distinguishable.
\end{theorem}

\begin{proof}
For ease of reference, recall that the parameters of $\Phi$ are: $g \geq 2$,  $\mi{i}$ and $\mathcal{Q}_i$ for $1\leq i \leq n$. The integer multiset $\mathcal{Q}_i$ has elements indexed by ${\bf J}=(J_1,\ldots,J_n)$  with $0\leq J_i <\mi{i}$ and has cardinality $v({\bf m})$.  
The least common multiple of all its elements is denoted by $Q_i$ and $M_i=g\mi{i} Q_i$, $M'_i=M_i/m_i$.

Suppose that two $\bf m$-blocks, $B({\bf x})$ tagged at ${\bf x}=(x_1,\ldots,x_n)$ and $B({\bf x}')$ at ${\bf x'}=(x'_1,\ldots,x'_n)$, have identical image sets under the mapping $\Phi_{\bf M}$.  
Without loss of generality, we can rotate the index on the cyclic grid so that we can assume ${\bf x'=0}$.   
By construction, a grid point in $B({\bf x})$ shares the same image under $\Phi_{\bf M}$ with a grid point in $B({\bf x}')$  if and only if they belong to the same sub-grids.

We label the sub-grids by means of index ${\bf J}$.  
For each sub-grid, there is a condition regarding the distance between $\bf x$ and $\bf 0$ that will be used to prove the theorem.

Recall that $\Phi_{\bf M}$ is a collection of color mappings $\Phi_{\bf J}$, one for each $\bf J$.   
Each $\Phi_{\bf J}$ is a product code on grid $G_{\bf M';J}^c$ defined by an $n$-dimensional repetitive code.   
In other words, for each $i, 1\leq i \leq n$, $\Phi_{\bf J}^{(i)}$ is built by repeating an $\mi{i}$-distinguishable code that has a base period denoted by $\ell_{\bf J}^{(i)}$.  
Furthermore, $\ell_{\bf J}^{(i)}=gq_{\bf J}^{(i)}$, where $q_{\bf J}^{(i)} \in \mathcal{Q}_i$.

First, consider the sub-grid $G^c_{{\bf 0}_n}$, which contains ${\bf x'=0}_n$.  
Let ${\bf y}=(y_1,\ldots,y_n)$ be the unique point in $B({\bf x})$ that also belongs to this sub-grid.  
We can represent $y_i$ by $\ell_i\mi{i}$ for some $\ell_i$ in $\mathbb{Z}_{M'_i}$.  
The assumption that two $\bf m$-blocks have identical mapping image sets implies that under the 1D code, $\Phi_{{\bf 0}_n}^{(i)}$, the points $0$ and $\ell_i$ in $\mathbb{Z}_{M'_i}$ are mapped to the same image.  
By Lemma~\ref{lemma:code distance}, $\ell_i$ is divisible by $\ell^{(i)}_{\bf J}=gq^{(i)}_{\bf J}$. 

In general, let ${\bf z_J}\in B({\bf x})$ and ${\bf J}\in B({\bf 0}_n)$ be any pair of points that belong to the same sub-grid, $G^c_{\bf J}$, where ${\bf J}=(j_1,\ldots. j_n)$, $0\leq j_i <\mi{i}$.
We can denote the $i$-th component of $\bf z_J$ by $z^{(i)}=r_i\mi{i}+j_i$ for some $r_i, 0\leq r_i <M'_i$.   
However, since $\bf y$ and $\bf z_J$ belong to the same $\bf m$-block, ${\bf z}\in B({\bf x})$, it follows that $r_i=\ell_i, \ell_i-1,$ or $\ell_i+1~(\!\!\!\mod M'_i)$.    

The assumption that two $\bf m$-blocks have identical mapping image sets implies that under the 1D code, $\Phi_{\bf J}^{(i)}$, $0$ and $r_i$ in $\mathbb{Z}_{M_i}$ are mapped to the same image.  
By Lemma~\ref{lemma:code distance}, $r_i$ is divisible by $\ell^{(i)}_{\bf J}=gq^{(i)}_{\bf J}$, hence $g$ divides $r_i$.
Since $g\geq 2$, this implies $r_i=\ell_i$.
So, $\ell_i$ is divisible by $gq^{(i)}_{\bf J}$ for all ${\bf 0}_n \leq \bf J < m$ and $\ell_i$ is divisible by $gQ_i=M'_i$.
Since $\ell_i < M'_i$, it follows that $\ell_i=0$ for all $i$.  
So ${\bf x}={\bf 0}_n$ and the two $\bf m$-blocks are identical.
\end{proof}

Note that the above result is simply a version of Theorem~\ref{theorem:construction} adapted to multi-dimensional  unitary braid code.  

\subsection{Unitary Braid Code on $n$-dimensional Grids of Arbitrary Size}\label{section:general braid code}

In this section, we explain how to define multi-dimensional unitary braid code for arbitrary size.

Given an $n$-dimensional flat grid, $G_{\bf M}$ of arbitrary size $\bf M$ with a block size $\bf m$, it is possible to define an $\bf m$-distinguishable color code on it by embedding it in a grid on which a standard braid code exists.  
It is strict forward to check that the restricted color mapping defines a distinguishable color code.

For a cyclic grid, $G^c_{\bf L}$, one can embed it in $G^c_{\bf M}$ such that $L_i \leq M_i$ for $1\leq i \leq n$. 
Let $\Phi$ be a color mapping that defines a standard unitary code on $G^c_{\bf M}$ as described in Section\ref{section:unitary braid code}.  The restriction of $\Phi$ to $G^c_{\bf L}$ defines a color code, however the
code may not be $\bf m$-distinguishable.   To obtain a distinguishable code, $\Phi$ may need to be modified. 
Below we present an approach that requires additional colors but is relatively easy to describe.

Recall that the unitary braid code defined by a color mapping $\Phi_{\bf M}$ is decomposed into $\nu(\bf m)$
product code mappings, $\Phi_{\bf J}$, which maps sub-grid $G^c_{\bf J}$ to a color set  $\mathcal{C}_{\bf J}$, for ${\bf 0}_{n} \leq  {\bf J} < {\bf m}$. Recall
$\mathcal{C}_{\bf J}$ is a product set, $\prod_i \mathcal{C}^{(i)}_{\bf J}$, with $|\mathcal{C}^{(i)}_{\bf J}|=gq^{(i)}_{\bf J}$.  
Construct a new collection of product color sets,
$\{ \mathcal{D}_{\bf J}\}$, as follows:  If ${\bf J}=(l,{\bf 0}^-_i)$ for some $i, 1\leq i \leq n$ and $l,0\leq l <m_i$, then a new color is added to $\mathcal{C}^{(i)}_{\bf J}$ so that
$D^{(i)}_{\bf J}=C^{(i)}_{\bf J}\cup \{ d^{(i)}_{\bf J} \}$
otherwise define $D^{(i)}_{\bf J}=C^{(i)}_{\bf J}$.  
Next, define new color product sets, $\mathcal{D}_{\bf J}=\prod_i \mathcal{D}^{(i)}_{\bf J}$ and 
let $\mathcal{D}=\bigcup_{\bf J}\mathcal{D}_{\bf J}$.  
We emphasize that in this construction of the new
color sets, we require that $d^{(i)}_{\bf J} \ne d^{(i)}_{\bf J'}$ if $\bf J \ne J'$.

We introduce a terminology to facilitate our subsequent
discussion -- A color in a color set is said to contain a factor color $c_i$ in the $i$-th coordinate if in the vector form, it can be represented as, $(c_i;{\bf c}^-_i)$.

Next, we describe an algorithm that allows us to construct
an $\bf m$-distinguishable color mapping by modification of
$\Phi_{\bf M}=\{ \Phi_{\bf J}\}$. The modified version of $\Phi$ maps to $\mathcal{D}$ and is defined iteratively on each coordinate, starting with $i=1$.  

Denote the modified output after the $k$-th iteration as $\Theta_{k}=\{ \Theta_{k;\bf J}\}$, which defines a color mapping on a cyclic grid, $G^c_{{\bf L}_k}$ where ${\bf L}_k={(L_1,...,L_k,M_{k+1},...,\ldots M_n)}$.
Moreover, just like $\Phi_{\bf J}$, $\Theta_{k;\bf J}$ is a product code, so it
has a structure:
\[
\Theta_{k;\bf J}({\bf x})=(\Theta^{(1)}_{k;\bf J}(x_1), \ldots, \Theta^{(n)}_{k;\bf J}(x_n))
\]
The mapping $\Theta_{k}$ then serves as the input to the $k+1$ iteration.  The iteration begins with $\Theta_{0}=\Phi$ and finishes with the output $\Psi \triangleq\Theta_{n}$, which is defined on the grid $G^c_{\bf L}$.

At the $i$-th iteration, we modify $\Theta_{i;\bf J}$ for each index ${\bf J}, {\bf 0}_{n} \leq  {\bf J} < {\bf m}$, as follows:

If $L_i=M_i$ or $L_i$ is not a multiple of $m_i$,  $\Theta_{i;\bf J}$ is simply a restriction of $\Theta_{i-1;\bf J}$ to $G^c_{{\bf L}_i}$.
Otherwise, if $L_i/m_i=R_i$ is an integer, $\Theta_{i;\bf J}$ is defined by restricting $\Theta_{i-1;\bf J}$ to $G^c_{{\bf L}_i}$ except for
${\bf J}$ of the form $(l;{\bf 0}^-_i)$ for some $l, 0\leq l <m_i$, the color mapping is then as follows:
\begin{enumerate}
\item For ${\bf x}=(R_i-1;x^-_i)\in G^c_{{\bf L}_i}$, define
\[
\Theta_{i;\bf J}^{(i)}({x_i})=d^{(i)}_{\bf J}, ~\Theta_{i;\bf J}^{(l)}({x_l})=\Theta_{i-1;\bf J}^{(l)}({x_l}) {\rm ~for~} l\ne i.
\]
\item Otherwise, define $\Theta_{i;\bf J}({\bf x})=\Theta_{i-1;\bf J}({\bf x})$.
\end{enumerate}

It is clear that $\Theta_{i;\bf J}$ is a product code as claimed.
Note that
\[
\mathcal{P}_i(\Psi({\bf x}))=
\mathcal{P}_i(\Theta_i({\bf x}))
\]
since the modifications at iteration steps at $l\ne i$  does not
affect the projected values into the $i$-th coordinate component of a color symbol.  If ${\bf x}=(lm_i+j,{\bf x}^-_i)$, then   
\begin{equation}\label{eqn:Psi}
\mathcal{P}_i(\Psi({\bf x}))=
\Psi^{(i)}_{(j,{\bf x}^-_i)}(l)=
\Theta^{(i)}_{i;(j,{\bf x}^-_i)}(l).
\end{equation}
We call a grid point ${\bf x}=(x_1, \ldots, x_n)$ of $G^c_{\bf L}$ exceptional in coordinate $i$, if $x_i$ satisfies $(R_i-2)m_i+1 \leq x_i <L_i=R_im_i$.  
The image set under $\Psi$ of  $B({\bf x})$, an $\bf m$-block at $\bf x$, contains a color that has a factor $d^{(i)}_{\bf J}$ in the $i$-th coordinate
for some $\bf J$ if and only if $\bf x$ is exceptional in coordinate $i$.

\begin{theorem}\label{theorem:general unitary}
$\Psi$, which is defined on $G^c_{\bf L}$, is $\bf m$-distinguishable if $L_i \geq 2m_i$ for all $i$.
\end{theorem}

\begin{proof}
Suppose that $\Psi(B({\bf x}))=\Psi(B({\bf y}))$.
We claim that for $1\leq i \leq n$, $x_i=y_i$.
Let ${\bf x}=(lm_i+j,{\bf x}^-_i)$, then according to Equation \ref{eqn:Psi},
\begin{align}
\mathcal{P}_i(\Psi(B({\bf x})))= 
\left(
\bigcup_{\overset{\scriptscriptstyle j\leq k <m_i,}{\scriptscriptstyle {\bf 0}_{n-1} \leq {{\bf J}^-_i} < {{\bf m}^-_i}}} \Psi^{(i)}_{(k;{\bf J}^-_i)}(l)
\right)
\bigcup
\left(
\bigcup_{\overset{\scriptscriptstyle 0\leq k <j,}{\scriptscriptstyle {\bf 0}_{n-1} \leq {{\bf J}^-_i} < {{\bf m}^-_i}}} \Psi^{(i)}_{(k;{\bf J}^-_i)}(l+1)
\right)
\end{align}
Note that $l+1$ in the above equation is modulo $R_i$.
Suppose $\bf x$ is exceptional at $i$.  If $l=R_i-2$, then $\Psi(B({\bf x}))$ contains colors with a factor $d^{(i)}_{(k;{\bf 0}^-_i)}$
for $k$, $0\leq k < j$, If $l=R_i-1$, then for 
$k$, $j \leq k < m_i$, $\Psi(B({\bf x}))$ contains colors with a factor
$d^{(i)}_{(k;{\bf 0}^-_i)}$.  
If $\bf y$ is not exceptional at $i$, then $\Psi(B({\bf y}))$ does not contain any color with a factor of the form $d^{(i)}_{\bf J}$ for any $\bf J$.  
So the two image sets cannot be identical unless $\bf y$ is also exceptional at $i$.  Let $y_i=l'm_i+j'$.  Note that unless $l=l'$ and $j=j'$ the two image sets cannot have identical collections of factors, 
$\{d^{(i)}_{(k;{\bf 0}^-_i)}\}$.  Hence, $x_i=y_i$.
 
Suppose $\bf x$ is not exceptional at coordinate $i$, then so must be $\bf y$.  If follows that
$\mathcal{P}_i(\Phi(B({\bf x})))=\mathcal{P}_i(\Psi(B({\bf x})))=\mathcal{P}_i(\Psi(B({\bf y})))=\mathcal{P}_i(\Phi(B({\bf y})))$.  Since $\Phi$ is $\bf m$-distinguishable, this implies  $x_i=y_i$ according to Proposition \ref{prop:unitary code}.   
\end{proof}

\section{Asymptotic Order of A Unitary Braid Code
}\label{section:order}

For grids with sizes approaching infinity, the efficiency of $n$-dimensional unitary braid code has the same asymptotic order as those constructions that require only minimal number of color symbols.
Before establishing this result, we first state a simple result on a product of $k$ consecutive primes. 

Let $r_k$ denote the $k$-th prime, starting with $r_1=2$.
The following lemma is probably known as folklore, but is included here for completeness.

\begin{lemma}\label{lemma:prime product}
Let $l,m$ be fixed positive integers and let $L$ be an integer satisfying $L\geq 2mr_1 r_2 \cdots r_{l}$. 
Denote by $\eta_L$ the smallest index such that $L\leq L_{l,m}\triangleq 2mr_{\eta_L}r_{\eta_L+1}\cdots r_{\eta_L+l-1}$.
Then
\begin{equation}\label{eq:order estimate}
O(L_{l,m})=O(L),
\end{equation}
and for $0 \leq i < l$,
\begin{equation}\label{eq:order estimate 0}
O(r_{\eta_L+i})=O(L^{1/l}).
\end{equation}
\end{lemma}
\begin{proof}
If $L<L_{l,m}$ then, $\eta_L > 1$. Hence,
\[
L_{l,m}-L \leq 2m[(r_{\eta_L} \cdots r_{\eta_L+l}) - (r_{\eta_L-1}\cdots r_{\eta_L+l-1})].
\]
Since, by the Bertrand-Chebyshev Theorem, 
\[
r_{\eta_L+l} < 2r_{\eta_L+l-1} < \cdots < 2^{l+1}r_{\eta_L-1},
\]
we have
\begin{align*}
0 < L_{l,m}-L  < 2m(2^{l+1}-1)r_{\eta_L-1} \cdots r_{\eta_L+l-1} <2m(2^{l+1}-1)L.
\end{align*}
Hence, $O(L)=O(L_{l,m})$.
Since $r_{\eta_L+l}< 2^{l}r_{\eta_L}$, $O(r_{\eta_L+i})=O(r_{\eta_L})$ for $0 \leq i < l$ and $O(r_{\eta_L+i})=O(L^{1/l})$. 

If $L=L_{l,m}$, \eqref{eq:order estimate} is trivial, and \eqref{eq:order estimate 0} can be proved by a similar argument as above.
\end{proof}
To avoid complicated notation, we focus on a cubical $n$-dimensional grid,  $G^c_{\bf M}$, with grid size,  ${\bf M}=(M, \ldots, M)$ and block size ${\bf m}=(m, \ldots, m)$.  
For $s \in \mathbb{Z}^+$, let $\mathcal{Q}_s \triangleq\{q_{\bf J}:{\bf 0}_n \leq {\bf J < m} \}$ denote the set of $m^n$ consecutive prime numbers beginning with $r_{s}$,
indexed by $\bf J$ in some order.  Let
$Q_s=\prod_{\bf J}q_{\bf J}=r_s\cdots s_{t+m^n-1}$, $M_s=2mQ_s$
and ${\bf M}_s=(M_s,\ldots, M_s)$.
Define a sequence of cyclic grids $G^c_{{\bf M}_s}$ for $s\in \mathbb{Z}^+$.
We have shown that there exists a standard unitary braid code on $G^c_{{\bf M}_s}$ that is $\bf m$-distinguishable. 

Let $K({\bf L, m})$ represent the minimal size of the color symbol set for the existence of an $\bf m$-distinguishable unitary braid code on a cubical $G^c_{\bf L}$.  
As guaranteed by the previously discussion in Section~\ref{section:unitary braid code}, there is a standard unitary braid code on $G^c_{{\bf M}_s}$, and hence
\begin{equation}\label{eq:order estimate 1}
K({\bf M}_s,{\bf  m}) \leq (2m)^nr_{s+m^n-1}^n.
\end{equation}
The upper bound can be made smaller by a more careful arrangement of the parameters used in the sunmao sub-grids, but this simple bound suffices for the asymptotic result we want to prove in this section.

Consider a cyclic grid $G^c_{\bf L}$, where
${\bf L}=(L,\ldots,L)$ with $L \geq 2mr_1\ldots r_{m^n}$.
Define $\eta_L$ as in Lemma \ref{lemma:prime product} by setting $l=m^n$.  
If $L=M_s$ for some $s$ then \eqref{eq:order estimate 1} holds and $\eta_L=\eta_{M_s}=s$.
Otherwise, by Theorem \ref{theorem:general unitary} we can embed $G^c_{\bf L}$ in $G^c_{{\bf M}_{s}}$ with $s=\eta_L$ and obtain a braid code on $G^c_{\bf L}$ by restriction or modification of a standard code.  Moreover,
\begin{equation}\label{eq:order estimate 2}
K({\bf L, m}) \leq m^n(2r_{\eta_L+m^n-1}+1)^n  < (2m)^n(r_{\eta_L+m^n-1}+1)^n <(4mr_{\eta_L+m^n-1})^n.
\end{equation}
\begin{theorem}\label{theorem: asym order}
As $L$ tends to infinity, $O(K({\bf L, m}))=O(L^{n/m^n})$.     
\end{theorem}

\begin{proof}
It follows from \eqref{eq:order estimate 1} and \eqref{eq:order estimate 2} that  
$O(K({\bf L,m}))=O(r^n_{\eta_L+m^n-1})$.   
The result then follows from~\eqref{eq:order estimate 0}.
\end{proof}
Let $\alpha$ be defined by $K^c_{\bf L}({\bf m})=(\nu({\bf L}))^\alpha$.  
It was shown in \cite{Paper1} that a natural lower bound of $\alpha$ is ${\nu({\bf m})}^{-1}$.
Hence, the asymptotic order of the lower bound of $K^c_{\bf L}({\bf m})$ for a cubical grid with volume $L^n$ is $O(L^{n/m^n})$, same as that of the unitary braid code.

\section{Decoding Algorithm for Braid Codes}
\label{section:decoding}
Decoding multiset codewords without relying on lengthy codebooks is an issue of practical and theoretical interest.
In~\cite{Paper1} an algorithm is presented for decoding 1D multiset codes with block size $m=2$.
In this section, we show an efficient decoding algorithm for braid code with arbitrary block size that does not require a codebook.

We first focus on a standard braid code defined by $\Phi$ over a 1D grid, $G^c_{M}$.  Let $J=M/m$.  
Hence, there are $J$ aligned blocks in $G^c_{M}$ and in $G^c_{M_i}$ for $1\leq i \leq I$, where $M_i=m_iM/m$ and $m=m_1+\cdots+m_I$. 
Note that the size of the generator of $\Phi_i$ is $\ell_i=gc_iq_i$ for some factor $c_i$ of $m_i$.

To facilitate our discussion, we introduce the concept of an {\it associated matrix} for a braid code as follows:

\begin{definition}
Given a braid code defined by $\Phi$,  the associated matrix, $\mathcal{A}$, is an $I$ by $J$ matrix, in which the $(i,j)$ entry, $\mathcal{A}_{i,j}$, is defined to be $\Phi_i(B_{m_i}(jm_i))$, the image set of the $m_i$-block of $G^c_{M_i}$ tagged at $jm_i$, the point that corresponds to 
$jm+d_i$ on $G^c_M$.
\end{definition}
Since the grid is cyclic, $\mathcal{A}$ should be regarded as a cyclic matrix in the sense that the last column of the matrix is followed by the first column. We summarize some useful properties of this matrix in the following lemma.

\begin{lemma}\label{lemma:associate matrix}
Matrix $\mathcal{A}$ satisfies the following properties:
\begin{enumerate}
\item The union set of all the entries in $\mathcal{A}$ is equal to the image set of $G^c_M$ under $\Phi$ and has multiset cardinality $M$.
\item The image set of all aligned $m_i$-blocks in $G^c_{M_i}$ under $\Phi_i$ is listed as entries in the $i$-th row of $\mathcal{A}$.
\item The image of the $j$-th aligned $m$-block of $G^c_M$ is equal to the union of all the entries of the $j$-th column of $\mathcal{A}$. 
\item \label{lemma:associate matrix 4}The image of an $m$-block at $jm+d_i$, $0\leq j <J$, $1 < i \leq I$, is the following union of $I$ matrix entries:
\[
\mathcal{A}_{i,j+1}\cup\mathcal{A}_{i+1,j+1}\cup\cdots\cup
\mathcal{A}_{I,j+1}\cup\mathcal{A}_{1,j+2}\cup\cdots\cup\mathcal{A}_{i-1,j+2}.
\] (The matrix entry indices, $i$ and $j$, start from 1 instead of 0 and the $J+1$ column corresponds to the first column.)
\item   \label{lemma:associate matrix 5}The image of an $m$-block at $jm+d_i+x_r$, $0\leq j<J$, $0<x_r<m_i$ for some $i, 1 \leq i \leq I$,  is equal to
\[
\mathcal{J}_i \cup
\mathcal{A}_{i+1,j+1}\cup\cdots\cup
\mathcal{A}_{I,j+1}\cup\mathcal{A}_{1,j+2}\cup\cdots\cup\mathcal{A}_{i-1,j+2}
\]
where $\mathcal{J}_i$ represents the image of the $m_i$-block of $G^c_{M_i}$ at $jm_i+x_r$ under $\Phi_i$ and an $m$-block of $G^c_M$ has at most one such image set not listed in $\mathcal{A}$

\item \label{lemma:associate matrix 6}The $i$-th row of $\mathcal{A}$ has a minimum period $gq_i$. 

\end{enumerate}
\end{lemma}
\begin{proof}
Statements (i), (ii) and (iii) follow by straightforward accounting, while statements (iv) and (v) follow from Lemma~\ref{lemma:sunmao 1D}(ii).

The first and the $j$-th aligned block of $G^c_{M_i}$ begin at 0 and $(j-1)m_i$ respectively.  By Lemma~\ref{lemma:code distance}(ii) if the $m_i$-blocks tagged at these two points have identical image sets, then their distance is divisible by $gm_iq_i$.  Hence, $gq_i$ divides $j$ and is the minimum period of row $i$.
\end{proof}
To streamline the notation, we label the multisets codeword of the $m_i$-blocks of $G^c_{M_i}$ according to the order they first appear in the braid code.

For the second braid code in Example~\ref{ex:braid_code}, the generator of the first sub-grid has size $\ell_1=3$ and the color mapping $\Gamma_1=a_1a_2a_3$.  
The multiset images are: $\{a_1,a_2\}$, $\{a_2,a_3\}$ and $\{a_1,a_3\}$, which are labeled as 0, 1, and 2 respectively.  
Similarly, the multiset images under $\Gamma_2=\Phi_2$, given in~\eqref{eq:ex_braid_case1-2}, are labeled in order as $0,1,\ldots,14$.
The associated matrix, $\mathcal{A}$, is 2 by 15 and its first nine columns are:
\begin{equation}\label{equation:associate matrix}
\mathcal{A}=
\begin{bmatrix}
0&1&2&0&1&2&0&1&2&\cdots\\
0&1&2&3&4&5&6&7&8&\cdots\\
\end{bmatrix}. 	
\end{equation}

Note that the generator of $G^c_{M_i}$ has size $\ell_i=gc_iq_i$ which may not be divisible by $m_i$
and row $i$ of $\mathcal{A}$ has a minimum period of $gq_i$. In other words, the period of this row 
is $m_i/c_i$ times of $\ell_i$ and is repeated $Q/q_i$ times in the matrix, where $Q=\lcm_i \{q_i\}$.  

Given $\mathcal{A}$, we construct a sub-matrix consisting of the
1st, $(g+1)$-th, $\ldots, ((J-1)g+1)$-th columns only.  
Let $\mathcal{B}$  be the $I$ by $Q=(J/g)$ integer matrix defined by
\[
\mathcal{B}_{i,j}=\mathcal{A}_{i,(j-1)g+1}/g.
\]
Lemma \ref{lemma:associate matrix 6}, follows directly from straightforward argument and well known
results in the literature \cite[Theorem 3.12]{Jones1998}.

\begin{lemma}
\label{lemma:associate matrix B}
For $i, 1\leq i \leq I$, one has
\[
\mathcal{B}_{i,j}=j {\rm ~mod~} q_i.
\]
Given a column vector of $\mathcal{B}$ , $(B_{1,j},\ldots,B_{I,j})^T$, the Generalized Chinese Remainder Theorem (CRT) allows the determination of $j$ by direct computation, which can be achieved with a complexity $O(I^3n_o^2)$,
where $n_0=\max_i \lceil {\log_2 q_i} \rceil$.
\end{lemma}

\subsection{Decoding 1D Braid Code}\label{section:decoding 1D}
We represent the multiset codewords as $k$-dimensional vectors in $\mathcal{Z}^k_m$, where $k$ is the cardinality of the color set. 
If there is a codebook-free decoding algorithm for the braid code, we can represent it as a function, $\rho$, 
from $\Phi(G^c_{M}) \subset \mathcal{Z}^k_m$ to $G^c_{M}$, so that, $\rho={\Phi}^{-1}$.  We denote the computational complexity of $\rho$ by $O(\rho)$.

For example, for $m=1$, a 1-distinguishable code requires $M$ colors and
a codeword is of the form $(1;{\bf 0}^-_i)$ for $1 \leq i \leq M$, and $\rho$ simply maps this vector to the grid point
$i-1$.   So, $O(\rho)$ is $O(M)$ if we search for the non-zero entry linearly.  
For $m=2$, a decoding algorithm is presented in \cite{Paper1} with $O(\rho)\leq O(\sqrt{M}\lceil\log_2 M \rceil)$.

Since a 1D braid code is synthesized from codes defined over sub-grids, a natural question
arises as to whether a decoding algorithm based on the
individual sub-grid decoding algorithms exists.  The following result addresses this question.

\begin{theorem}\label{theorem:decoding 1D}Given a sunmao decomposition, $m=m_1+\ldots +m_I$, let $\Phi$ define a braid code on $G^c_M$  with $J=M/m \in \mathbb{Z}^+$ with color set $\mathcal{C}=\cup_i \mathcal{C}_i$.
Denote the image set by $\mathcal{S}=\Phi(G^c_{\bf M}) \subset \mathbb{Z}^{|\mathcal{C}|}$. Let $\rho_i$ be a decoding algorithm on $G^c_{M_i}$ for $M_i=m_iM/m$ and $1\leq i \leq I$. Then there exists a decoding algorithm $\rho : \mathcal{S} \rightarrow \mathbb{Z}_J\times\mathbb{Z}_m$ with complexity $O(\rho)\leq O(n_0^2)+\sum_i O(\rho_i)$, assuming parameters $m$ and $g$ are fixed,
\end{theorem}

\begin{proof}
In vector representation, a codeword, $\bf c$, is a sequence of $I$ sub-vectors, the $i$-th sub-vector, ${\bf c}_i$, lists the elements in $\mathcal{C}_i$ in the codeword. By applying the decoding algorithm $\rho_i$ to  ${\bf c}_i$, we obtain a unique grid point on $G^c_{\ell_i}$, the generator of the repetitive code on $G^c_{M_i}$.  (The image of the $m_i$-block tagged at that point under $\Phi_i$ is ${\bf c}_i$.) However, there may be more than one  such point in $G^c_{M_i}$.

We refer to the list of partially decoded points on $G^c_{\ell_i}$'s as the sunmao list and represent the $i$-th element as an ordered pair $(j_i,r_i)$, so that $j_im_i+r_i$, $0\leq j_i <gq_i, 0\leq r_i <c_i$, is an element of  $G^c_{\ell_i}$.

According to Lemma \ref{lemma:associate matrix}, there are 3 types of possible structures that can appear in the sunmao list:
\begin{enumerate}
 \item   There is a single column in $\mathcal{A}$, the $(j^*+1)$-th column, so that $\mathcal{A}_{i,j^*}=j_i$ for all $i$.   This corresponds to the image of the $(j^*+1)$-th aligned $m$-block of $G^c_M$ that begins at the grid point $j^*m$.
\item There are two adjacent columns in $\mathcal{A}$,  the $(j^*+1)$-th and $(j^*+2)$-th column,  in which $j_i$ appears as an entry in the following way:  There exists as $i^*, 1\leq i^* \leq I$, such that for $i^* \leq i \leq I$, $j_i$ appears in the $(j^*+1)$-th column of $\mathcal{A}$  and for $1\leq i< i^*$, $j_i$ appears in the $(j^*+2)$-th column of $\mathcal{A}$.  This corresponds to an $m$-block of $G^c_M$ that begins at $j^*m+d_{i^*}$. 
\item There is one and only one entry, $i^*$, such that $r_i \ne 0$.  For $i^*< i \leq I$, $j_i$ appears in the $(j^*+1)$-th column of $\mathcal{A}$  and for for $1\leq i< i^*$, $j_i$ appears in the $(j^*+2)$-th column of $\mathcal{A}$.  This corresponds to an $m$-block of $G^c_M$ that begins at $j^*m+d_{i^*}+r_{i^*}$.
\end{enumerate}
 
If we can determine the triple, $(j^*,i^*,r_{i^*})$, then  $\rho({\bf c})=j^*m+d_{i^*}+r_{i^*}$ defines the decoding output.

To determine $r_{i^*}$ we simply scan the sunmao list for a non-zero entry.  If none exists then $r_{i^*}$ is set to 0.
This step can be performed with computation complexity $O(n_0)$ for fixed $m$ and $g$.

It is easy to determine $i^*$ if there is a non-zero $r_i$ entry.  Next, we focus on the case that all $r_i=0$. 
Since every row of $\mathcal{A}$ has a period that is a multiple of $g$ and $g>1$,  hence the two-column structure mentioned above is sandwiched between two adjacent columns of the matrix $\mathcal{B}$.
We label these columns in $\mathcal{B}$ as $(B_{1,a^*+1},\ldots,B_{I,a^*+1})^T$ and $(B_{1,a^*+2},\ldots,B_{I,a^*+2})^T$, 
where $j^*=a^*g+b^*$, $0\leq a^* <Q, 0\leq b^*<g$. 

Note that for the sub-grid $I$,
\[
j_I=j^* {\rm ~mod~} gq_I,
\]
and for the other sub-grids, $j_i$ either satisfies
\begin{equation}\label{theorem: decode index 1}
j_i=j^* {\rm ~or~}  j^*+1 {\rm ~mod~} gq_i.
\end{equation}
The smallest $i$ for which the equation
\[j_i=j^* {\rm ~mod~} gq_i
\]
holds defines $i^*$, as justified by the 3 listed structures.  
For each index $i$, we want to decide
which case would hold in Equation \ref{theorem: decode index 1}.

Note that for $1 \leq i \leq I$, $j_i$ has a representation in the form
\begin{equation}\label{theorem:decoding index 2}
j_i=a_ig+b_i,
\end{equation}
$0\leq a_i <q_i, 0\leq b_i <g$.   Since $j^*=j_I+sgq_I$ for some $s \in \mathbb{Z}$,
\[
b^*=j^* {\rm ~mod~} g=j_I {\rm ~mod~} g=b_I.
\]
Similarly one can show that $j_i=j_I$ or $j_I+1 {\rm ~mod~} g$.  Note that the two cases are distinct since $g>1$.
So by computing $(a_i,b_i)$ by means of Equation \ref{theorem:decoding index 2} and by comparison with $b_I$, one can determine $i^*$.   This step can be performed with computation complexity $O(n_0)$.

If $b_I=g-1$ and $i^*>1$, the vector $(a_1,\ldots, a_I)^T$
consists of entries of two adjacent columns of $\mathcal{A}$.
Hence, for  $i, 1\leq i < i^*$, redefine $a_i$ to be 
$(a_i-1)  {\rm ~mod~} q_i$, and  no adjustments for other $a_i$'s
are required.
The modified column vector $(a_1,\ldots, a_I)^T$ 
is equal to the column vector  $(A_{1,ga^*+1},\ldots,A_{I,ga^*+1})^T$.  
This step can be performed with computation complexity $O(n_0)$.

The final step of the algorithm aims to determine $a^*$ from the modified vector $(a_1,\ldots, a_I)^T$,
which corresponds to the $(a^*+1)$-th column of $\mathcal{B}$.   Hence, $a^*$ can be determined with computation complexity no more than $O(n^2_0)$ by means of the Generalized CRT.

 With the determination of $a^*$, $b^*=b_I$, $j^*=a^*g+b^*$, we then set $\rho({\bf c})=j^*m+d_{i^*}+r_{i^*}$.
\end{proof}

For unitary braid code, the above procedure can be greatly simplified.  It is worthwhile to briefly restate the
algorithm steps when $g=2$.  In this case, entries of $\mathcal{A}$ have entries that are either even or odd and
$\mathcal{B}$ consists of column vectors with even value entries.  Since $r_i$ is always 0, we only need to determine
$j^*$ and $i^*$.  If the entry $j_I$ is even valued, then, the smallest $i, 1\leq i \leq I$ for which $j_i$ is even is equal to $i^*$.
For $1 \leq  i <i^*$, redefine $a_i$ as $a_i-1 {\rm ~mod~} q_i$, then the newly defined vector $(a_1,\ldots, a_I)^T$ corresponds to $(B_{1,a^*+1},\ldots,B_{I,a^*+1})^T$.

We provide a simple example to illustrate the unitary case.  Let $m=2, I=2, g=2$, $q_1=2, q_2=3$, 
then $M=24$.  The associate matrix of this braid code is:

\begin{equation}
\mathcal{A}=
\begin{bmatrix}
\color{blue}
{0 ~1~2 ~3 ~0 ~1 ~2 ~3~ 0~ 1~ 2 ~3}\\
\color{purple}
{0 ~1 ~2 ~3 ~4~ 5 ~0 ~1 ~2~3 ~4~ 5}\\
\end{bmatrix}. 	
\end{equation}

Notice that the columns of $\mathcal{A}$ are composed of either odd or even-valued entries, which allows determination of $i^*$.  Set 
\begin{equation}
\mathcal{B}=
\begin{bmatrix}
\color{blue}
{0 ~1~0 ~1 ~0 ~1 }\\
\color{purple}
{0 ~1 ~2 ~0 ~1~ 2 }\\
\end{bmatrix}. 	
\end{equation}
Codewords $\{\color{blue}{2},\color{purple}{2} \}$, $\{\color{blue}{3},\color{purple}{2} \}$,$\{\color{blue}{3},\color{purple}{3} \}$ and $\{\color{blue}{0},\color{purple}{3} \}$ all use the column, $(\color{blue}{1},\color{purple}{1})^T$, in $\mathcal{B}$ as input to the Generalized CRT algorithm.

 It is also of interest to note that the above coding defines a bijective mapping between $\mathbb{Z}_{24}$
 and $\mathbb{Z}_4 \times \mathbb{Z}_{6}$.  Hence, our proof of distinguishable color code can be regarded as yet another approach to generalize CRT.

If $\Phi$ defines a non-standard braid code on $G^c_M$, we
can embed $G^c_M$ to a cyclic grid on which a standard code
exists.  If $\Phi$ is a restriction of a standard braid code on the super-grid, the previously stated algorithm can be used
to decode codewords of $\Phi$.  If $\Phi$ is obtained by modifying
a standard braid code as described in Section \ref{section: 1D generalized braid}, then codewords
for grid points from $M-2m+1$ to $M-1$, need
to be screened and decoded separately.  These codewords
are characterized by the fact that they contain more than
1 symbol element from certain sub-grid color mapping.
The complexity of these additional steps is at most linear in $n_0$.

 \subsection{Decoding $n$-Dimensional Unitary Braid Code}\label{section:decoding nD}
There is also an efficient algorithm to decode multi-dimensional unitary braid code.
Let $\Phi$ define such a standard braid code on 
$G^c_{\bf M}$ with block size $\bf m$.
Let $B({\bf x})$ be the $\bf m$-block tagged at ${\bf x}=(x_1,\ldots, x_n)$.
According to Proposition \ref{prop:unitary code}, it is possible to decode $x_i$ from $\mathcal{P}_i(\Phi(B({\bf x})))$.

The braid code consists of $\nu({\bf m})$ product codes indexed by $\bf 0 \leq J < m$. 
Recall that $\mathcal{Q}_i \triangleq\{q^{(i)}_{\bf J}:{\bf 0}_n \leq {\bf J < m} \}$ is a collection of strictly positive integers with the lowest common multiplier denoted by $Q_i$, and $M_i=m_igQ_i$.
For each $\bf J$, $\Phi^{(i)}_{\bf J}$ defines a mapping on $G^c_{M_i/m_i}$, which can be represented as a row of integers, each representing a color multiset, (a singleton in this case,) as explained in the previous subsection.  The sequence has a minimum period of $gq^{(i)}_{\bf J}$.

Define an $\nu({\bf m})\times M_i/m_i$ matrix, $\mathcal{D}=(\mathcal{D}_{i,j})$, by mapping each of these $\nu({\bf m})$ rows of integers in the following way.
Let $w=\nu({\bf m})/m_i$, and divide the rows into $m_i$ consecutive bands, each band contains $w$ rows
starting from band 0.  The row defined by $\Phi^{(i)}_{(r;{\bf J}^-_i)}$ is mapped to the $r$-th band, where $0 \leq r <m_i$,
in some arbitrary order.

Then according to Equation \ref{eq:nD braid code}, the image set of the $j$-th aligned $m_i$-block of $G^c_{gQ_i}$ is
equal to the union of the entries of the $j$-th column of $\mathcal{D}$  and the image of the block beginning at $jm_i+r$, $0\leq r < m_i$ is equal to 
\[
\left(
\bigcup_{rw \leq s < \nu({\bf m})}\mathcal{D}_{s+1,j+1} \right) \bigcup
\left( \bigcup_{0 \leq s < rw}\mathcal{D}_{s+1,j+2} \right).
\]
On the other hand, matrix $\mathcal{D}$ can also be regarded as the associate matrix of a unitary braid code constructed on the
1D grid, $G^c_{\nu({\bf m})gQ_i}$ with block size $\nu({\bf m})$.
Hence, $x_i$ can be uniquely determined by the method presented in the previous sub-section with complexity
$O(n_0^2)$, where $n_0=\max_{i,\bf J} \lceil {\log_2 q^i_{\bf J}} \rceil$.

\begin{theorem}
There exists a decoding algorithm for $n$-dimensional unitary
braid code with computation complexity $O(n^2 _0)$.
\end{theorem}
\begin{proof}
For a standard braid code, the argument is shown
already.  For a non-standard braid code obtained as
a restriction of a standard code, it is obvious that a decoding
algorithm with complexity $O(n^2_0)$ exists.  If the code
is obtained by adding colors to a standard code
as described in Section \ref{section:general braid code}, then the codewords at the exceptional points need to be
handled separately.  The addition steps can be done
with computation complexity at most $O(n_0)$.
\end{proof}
\section{Error Correction Property of 1D Braid Code}
\label{section:non-product}

A multiset code consists of $m$ color symbols, but in actual operation some color symbols may be missing.  Is it possible to decode the source information in this case?   

Let $G^c=\mathbb{Z}_M$ be a cyclic grid on which a color mapping, $\Phi$, is defined. Consider the case where a multiset code with $e$ missing symbols has been received; in other words, we have to decode based only on $m-e$ color symbols.  If the received symbols are in the form:
\[
\{ c(x_i),c(x_{i+1}),\ldots, c(x_{i+m-e-1})\}
\]
for some $x_i,\ldots,x_{i+m-e+1} \in G^c$ (the arithmetic used in the index calculation is modulus $M$), the complete multiset code could be one of the following $e$ choices:
\[
\begin{array}{c}
\{c(x_{i-e}),c(x_{i-e+1}),\ldots, c(x_{i+m-e-1})\},\\
\{c(x_{i-e+1}),c(x_{i-e+1}),\ldots, c(x_{i+m-e})\},\\
\vdots\\
\{c(x_{i}),c(x_{i-e+1}),\ldots, c(x_{i+m-1})\}.
\end{array}
\]
Therefore, no matter how the color mapping $c$ is defined, it is not possible to determine the exact location of the $m$-block in this case.  However, we note that the distance between tag points of  these blocks is at most $e$.  Hence, it may be possible to decode the source with a resolution
less than $e$, that is, the distance between possible 
solution candidates are at most at a distance $e$ from each other. On the other hand, it is not clear that
it is always possible to decode within such a resolution
for a given value $e$.

\begin{definition}
Consider a cyclic grid $G^c=\mathbb{Z}_M$ with block $m$.
For $e,1 \leq e <m$, a multiset color code for $G^c$ is correctable with resolution $e$, if given any multi-subset of the image set of an $m$-block tagged at $x \in G^c $ with cardinality $m-e$, it is possible to determine $x$ with resolution up to $e$.
\end{definition}

Consider a 1D grid $G^c_M$ with a braid code $\Phi$ defined with block size $m$, $g\geq 2$ , $M=gmQ$, where $Q$ is the $lcm$ of $\{q_1,\ldots,q_m\}$.  Moreover, the repetitive code on sub-grid $i$, which has size $M/m$ is generated by a 1-distinguishable code defined on $G^c_{p_i}$ where $\ell_i=gq_i$.

For $i, 1\leq i <m$, let $\hat{q}_i$ denote the minimum of products formed by subsets of $Q$ with $i$ elements.  

\begin{theorem}
If we restrict $\Phi$ to any sub-grid $\mathbb{Z}_{M'}$ with $M' \leq gm\hat{q}_{m-e}$, then the restricted code on $\mathbb{Z}_{M'}$ is correctable within resolution $e, 0\leq e <m$.
\end{theorem}

\begin{proof}
Consider two sets $\mathcal{S}$ and  $\mathcal{S}'$ with $m-e$ elements each that are subsets of respective $m$-blocks in $G^c_M$ .  By rotating the index notation, we can assume without lost of generality that
\[
S=\{x_1,\ldots, x_{m-e} \}
\]
with $x_1=0$.  If $S$ and $S'$ have identical images under $\Phi$, 
then their elements can be paired by the sub-grid they belong to.  Let $x'_i$ be the element in $S'$ that pairs with $x_i$.
Then $x'_i=j_im+x_i$ for some $j_i$.  The fact $x'_i$'s belong
to the same $m$-block implies that $|j_i-j_1| \leq 1$.
Since $x'_i-x_i$ is divisible by $gq_i$, this implies
$j_i=j_1$ for all $i$ since $g>2$.,  It follows then that
$x'_1$ is divisible by all  $i \in \{1,\ldots, m-e\}$.
So if $\Phi$ is restricted to $G^c_{M'}$ this implies $x'_1=0$.
If $x_{m-e}=m-1$,  then the $m$-block containing $\mathcal{S}$ is unique and is equal to $\{0,\ldots,m-1\}$.  
If $x_{m-e}=m-i$, $1<i\leq e+1$, then there are $i$ 
candidate $m$-blocks that can serve as the superset of $\mathcal{S}$, and the maximum distance before two candidate solutions is $e$.
\end{proof}


\section{Conclusion and Open Problems}
\label{section:conclusion}

Multiset color code can serve as a natural data representation format for many applications.  Object tracking in a multi-dimensional sensor grid offers an immediate application opportunity.  
Another application is a light based positioning system as presented in \cite{OGC23}.
For more examples, consider a chain of connected parts that can be tagged by markers at regular intervals; one can envision a DNA strand of macromolecules for illustration.  When a segment of the chain, say consisting of a block of $l$ markers, is broken up into the component macromolecules, the ordering of the markers is destroyed.   However, the multiset of markers may provide useful information to locate the broken site of the segment.

By conceptualizing the set of markers as a multi-dimensional grid.
we detailed a methodology of designing multiset color codes by first decomposing a grid into sub-grids, obtain color multiset coding  solutions for the sub-grids, and then piece the solutions together to
serve the original grid.  We refer the approach as {\it sunmao construction}.  To further concretize the idea, we propose a class of codes, referred to as braid codes, first for 1-dimensional grids and then extend to multi-dimensional grids by means of product code.  We can measure coding efficiency by the number of distinct colors need to uniquely identify all blocks of a fixed size in a grid.
We prove that the coding efficiency of braid codes has asymptotic  order  equal to that of the optimal codes as the grid size tends to infinity,   We also show that fast
decoding algorithms exist for braid codes and
1-dimensional braid code has inherent error correction properties.

Despite the results presented, there are many challenging problems that remain to be addressed in color multiset coding.  We mention two areas:

First of all, although braid codes can achieve asymptotically optimal order, the problem of finding color multiset codes that require the minimal number of distinct colors for a grid of a given size is an interesting and difficult problem.  Even for low dimension grids, such as 1D or 2D grids, these investigations will lead to many new combinatorial challenges,

Second, construction of error correction color codes is of great theoretic and practical interest.  Much work remains to be done in this direction. 


\appendices

\section{Proof of Lemma~\ref{lemma:sunmao 1D}}
\label{appendix:sunmao 1D}

For $0\leq j < M/m-1$, the $m$-block $B$ consists of grid points $\{y\in \mathcal{Z}_M: jm+d_i+x_r \leq y <(j+1)m+d_i+x_r \}$,
and for $j=M/m-1$, $\{y\in \mathcal{Z}_M: jm+d_i+x_r \leq y \leq jm+m-1 \}\cup
\{y\in \mathcal{Z}_M:0\leq y <d_i+x_r \}$.

If $x_r=0$ and $i \leq l \leq I$, then $B\cap S_l=\{y\in \mathcal{Z}_M: jm+d_l \leq y  \leq jm+d_{l+1}-1 \}$; if $1 \leq l <i$, then $B\cap S_l=\{y\in \mathcal{Z}_M: (j+1)m+d_l \leq y \leq (j+1)m+d_{l+1}-1 \}$.
Statement (ii)a and (ii)b then follow.

If $x_r>0$, then $B\cap S_i=\{y\in \mathcal{Z}_M: jm+d_i+x_r \leq y \leq jm+d_{i+1}-1 \}
\cup \{y\in \mathcal{Z}_M: (j+1)m+d_i \leq y <(j+1)m+d_i+x_r \}$.  
Under $\theta_i$, this set maps to $\{y \in \mathcal{Z}_{M_i}:jm_i+x_r \leq y \leq jm_{i}+m_i-1 \} \cup
\{y \in \mathcal{Z}_{M_i}: (j+1)m_i \leq y < (j+1)m_i+x_r \}$,
which is the non-aligned $m_i$-block at $jm_i+x_r$ in $G^c_{M_i}$.

For $i<l \leq I$,
$B\cap S_l=\{y \in \mathcal{Z}_M: jm+d_l \leq y \leq jm+d_{l+1}-1 \}$,
which is mapped to the $m_l$-block at $jm_l$, an aligned block.

For $1 \leq l < i$,
$B\cap S_l=\{y \in \mathcal{Z}_M: (j+1)m+d_l \leq y \leq (j+1)m+d_{l+1}-1 \}$,
which is mapped to the $m_l$-block at $(j+1)m_l$, also
an aligned block.  Statement (ii)c holds.  
Statement (i) then follows.

Finally, suppose $x'=j'm+d_{i'}+x'_r$ is a distinct grid point from $x$ with the same sub-block decomposition.  
The two corresponding sets of sub-blocks cannot be identical unless $j=j'$ and $i=i'$. 
But then $x_r \ne x'_r$ and the blocks at $x$ and $x'$ of $G^c_{M_i}$ are different, a contradiction. 
\qed


\section{Proof of Proposition~\ref{prop:braid-code-color-number}}
\label{appendix:braid-code-color-number}
Suppose $\ell_i>\ell_j$ for some $1\leq i<j \leq I$.
Since $\boldsymbol{\ell}$ is canonical, this occurs only when $m_i\neq m_j$, and there exists an index $t$ with $i\leq t<j$ such that $\ell_t>\ell_{t+1}$.
As $m_I\leq 3$, 
It suffices to show that
\begin{equation}\label{eq:canonical}
K^c_{\ell_t}(\alpha)+K^c_{\ell_{t+1}}(\beta)>K^c_{\ell_{t+1}}(\alpha)+K^c_{\ell_{t}}(\beta),
\end{equation}
for $(\alpha,\beta)=(1,2), (1,3)$ and $(2,3)$.

We need the exact values of $K^c_\ell(m)$ for $m=1,2,3$.
Let $M^c_m(k)\triangleq\max\{\ell:\,K^c_\ell(m)\leq k\}$ denote the maximum cyclic 1D grid size achievable using $k$ symbols subject to the $m$-distinguishable condition.
The values of $M^c_{m}(k)$ for $1\leq m\leq 3$ have been established in~\cite[Section III]{Paper1}:
\begin{equation}\label{eq:largest-length-1}
    M^c_1(k)=k
\end{equation}
for $k\in\mathbb{Z}^+$;
\begin{equation}\label{eq:largest-length-2}
    M^c_2(k)=\begin{cases}
        \binom{k+1}{2} & \text{if }k\equiv 1\bmod 2,\\
        \binom{k+1}{2}-\frac{k}{2} & \text{if }k\equiv 0\bmod 2;
    \end{cases}
\end{equation}
and
\begin{equation}\label{eq:largest-length-3}
    M^c_3(k)=\begin{cases}
        \binom{k+2}{3} & \text{if }k\equiv 1,2\bmod 3,\\
        \binom{k+2}{3}-\frac{k}{3} & \text{if }k\equiv 0\bmod 3.
    \end{cases}
\end{equation}
Note that each of the three sequences $M^c_1(k), M^c_2(k)$ and $M^c_3(k)$ is strictly increasing.

First, consider the case when $(\alpha,\beta)=(1,2)$.
For notational convenience, let $k_1=K^c_{\ell_t}(1), k'_1=K^c_{\ell_{t+1}}(1), k_2=K^c_{\ell_t}(2)$ and $k'_2=K^c_{\ell_{t+1}}(2)$.
We aim to claim that $\left(k_1-k'_1\right) - \left(k_2-k'_2\right)>0$.
By assumption, one has $k_1\geq k'_1$ and $k_2\geq k'_2$.
It follows from~\eqref{eq:largest-length-1} that $k_1=\ell_t$ and $k'_1=\ell_{t+1}$.
When both $k_2$ and $k'_2$ are odd, by~\eqref{eq:largest-length-2}, we have 
\begin{align*}
\ell_{t}-\ell_{t+1} = \binom{k_2+1}{2}-\binom{k'_2+1}{2} = \sum_{i=k'_2+1}^{k_2}\binom{i}{1} = \frac{1}{2}(k_2+k'_2+1)(k_2-k'_2),
\end{align*}
which yields $k_2-k'_2=\frac{2}{k_2+k'_2+1}(\ell_t-\ell_{t+1})$.
Therefore, 
\begin{align*}
    \left(k_1-k'_1\right) - \left(k_2-k'_2\right) = \frac{k_2+k'_2-1}{k_2+k'_2+1}\left(\ell_t-\ell_{t+1}\right)>0.
\end{align*}
When $k_2$ is odd and $k'_2$ is even, it follows from~\eqref{eq:largest-length-2} that
\begin{align*}
    \ell_{t}-\ell_{t+1} &= \binom{k_2+1}{2}-\binom{k'_2+1}{2}+\frac{k'_2}{2} \\
    &= \frac{1}{2}\left((k_2+k'_2)(k_2-k'_2)+k_2\right) >\frac{1}{2}(k_2+k'_2)(k_2-k'_2), 
\end{align*}
which yields $k_2-k'_2<\frac{2}{k_2+k'_2}(\ell_t-\ell_{t+1})$.
As $k_2+k'_2\geq 3$ in this case, 
\begin{align*}
    \left(k_1-k'_1\right) - \left(k_2-k'_2\right) = \frac{k_2+k'_2-2}{k_2+k'_2}\left(\ell_t-\ell_{t+1}\right)>0.
\end{align*}
Thus, the proof for $(\alpha,\beta)=(1,2)$ is complete, since the other two cases follow by the same argument.

Then, consider the case when $(\alpha,\beta)=(1,3)$.
Let $k_3=K^c_{\ell_t}(3)$ and $k'_3=K^c_{\ell_{t+1}}(3)$.
We aim to claim that $\left(k_1-k'_1\right) - \left(k_3-k'_3\right)>0$.
When both $k_3$ and $k'_3$ are not multiples of $3$, by~\eqref{eq:largest-length-3}, we have 
\begin{align*}
\ell_{t}-\ell_{t+1} = \binom{k_3+2}{3}-\binom{k'_3+2}{3} = \sum_{i=k'_2+1}^{k_2}\binom{i}{2} > \sum_{i=k'_2+1}^{k_2}\binom{i}{1}.
\end{align*}
The rest of the proof reduces to the first case.

Finally, consider the case when $(\alpha,\beta)=(2,3)$.
Similarly, to simplify the presentation, we only consider the case that $k_2,k'_2$ are odd and $k_3,k'_3$ are not multiples of $3$, since the other cases can be dealt with in the same way.
In this case, by~\eqref{eq:largest-length-2}--\eqref{eq:largest-length-3}, we have $\binom{k_2+1}{2}=\ell_t=\binom{k_3+2}{3}$ and $\binom{k'_2+1}{2}=\ell_{t+1}=\binom{k'_3+2}{3}$.
Starting from the identity $\binom{n+1}{2}=\binom{m+2}{3}$, we simplify it to $3n^2+3n=m(m+1)(m+2)$, and consequently obtain $n=\frac{1}{2}\Big(-1+\sqrt{1+\frac{4}{3}m(m+1)(m+2)}\Big)$, which implies that $n-m = \sqrt{\frac{1}{3}\left(m+3+\frac{2}{m}+\frac{3}{4m^2}\right)}-\frac{1}{2}$.
Define a function 
\begin{align*}
    f(x)\triangleq \sqrt{\frac{1}{3}\left(x+3+\frac{2}{x}+\frac{3}{4x^2}\right)}-\frac{1}{2}.
\end{align*}
From the first derivative of $f(x)$, we can easily verify that $f(x)$ is increasing when $x>1$.
Since $k_3>k'_3$ by assumption, it follows that
\begin{align*}
    k_2-k_3 = f(k_3) > f(k'_3) = k'_2-k'_3,
\end{align*}
as desired. 
\qed

\section{Proof of Proposition \ref{prop:restriction}}
\label{appendix:restriction}

Let $C_i$ denote the color set of the $i$-th sub-grid under $\Phi$.  To facilitate our subsequent discussion, we introduce a characteristic word of length equal to the grid size, $M$.  The word summarizes information on $\Phi$ so that its $j$-th letter is $C_i$ if the $j$-th grid point maps to $C_i$.  Note that the word should be interpreted cyclically so that its first and last alphabet symbols are considered consecutive.

Since $m$ does not divide $M_r$, we can write $M_r$ as $lm+d$, with $0<d<m$.  
The unitary braid code restricted to $G^c_{lm+d}$ yields the following word:
\[
\underbrace{
C_1 C_2 \cdots C_m \cdots  C_1 C_2 \cdots C_m}_\text{{\it lm} letters}
C_1 C_2 \cdots C_d. 
\]
Given a word fragment with $m$ consecutive letters, we refer to its corresponding multiset as an {\it $m$-fragment multisets}.  Note that if two $m$-fragment multisets are distinct, the color codewords of corresponding $m$-blocks must be different.

For $i$ ranging from 1 to $M_r-m+1$ , the multisets of $m$-fragment starting at the $i$-th grid points are $\{C_1, \ldots, C_m \}$, which can be represented by the $m$-dimensional vector,  ${\bf 1}_m=(1, 1, \ldots, 1)$.  In this representation, the $j$-th entry denotes the number of set elements from $C_j$.  To proof the proposition, we only need to check the images set of $m$-blocks beginning at grid point $(M_r-m+1+j)$ for $j$  ranging from 1 to $m-1$ are distinct.  We claim that the corresponding $m$-fragment multisets are distinct.

First consider the case $d\leq m-d$. We list the $m-1$ fragment multisets below.  There are $d$ words of the form
\begin{align*}
\{C_d, C_1, \ldots, C_{m-1}\}, \{C_{d-1}, C_d, C_1, \ldots, C_{m-2}\}, \ldots, \{C_1,\ldots, C_d, C_1, \ldots, C_{m-d} \},
\end{align*}
followed by $m-2d$ words of the form
\begin{align*}
\{C_m, C_1,\ldots, C_d, C_1, \ldots, C_{m-d-1} \}, \ldots, \{C_{2d+1}, \ldots,C_m, C_1,\ldots, C_d, C_1, \ldots, C_{d}\},
\end{align*}
and finally $d-1$ words of the form
\begin{align*}
\{C_{2d}, \ldots,C_m, C_1,\ldots, C_d, C_1, \ldots, C_{d-1} \}, \ldots, \{C_{d+2}, \ldots, C_m,C_1,\ldots, C_d, C_1 \}.
\end{align*}
The vector representation for the first $d$ multisets:
\begin{align*}
&(\underbrace{1,1,...1,2}_\text{\it d},1,\ldots,1,\underbrace{1,1,\ldots,1,0}_\text{\it d}), (\underbrace{1,1,...2,2}_\text{\it d},1,\ldots,1,\underbrace{1,\ldots 1,0,0}_\text{\it d}), \\ 
&\ldots \ldots, (\underbrace{2,2,...,2}_\text{\it d},1,\ldots,1,\underbrace{0,\ldots,0}_\text{\it d});
\end{align*}
for the second group:
\begin{align*}
&(\underbrace{2,2,...,2}_\text{\it d},1,\ldots,1,0,\underbrace{0,\ldots,1}_\text{\it d}), (\underbrace{2,2,...,2}_\text{\it d},1,\ldots,1,0,0,\underbrace{0,\ldots,1,1}_\text{\it d}), \\
&\ldots\ldots, (\underbrace{2,2,...,2}_\text{\it d},\underbrace{0,\ldots,0}_\text{\it d},\underbrace{1,\ldots,1}_\text{\it m-2d});
\end{align*}
and, for the third group:
\begin{align*}
&(\underbrace{2,...,2,1}_\text{\it d},\underbrace{0,\ldots,0,1}_\text{\it d},1,\ldots,1), (\underbrace{2,...2,1,1}_\text{\it d},\underbrace{0,\ldots,0,1,1}_\text{\it d},1,\ldots,1),\\
&\ldots\ldots, (\underbrace{2,1,...,1}_\text{\it d},0,1,\ldots,1).
\end{align*}
It is clear that these vectors are distinct and not equal to ${\bf 1}_m$.

Suppose $d> m-d$.  Again we only need to check the image sets of $m$-blocks beginning at grid point $(M-m+1+j)$ for $j$  ranging from 1 to $m-1$ are distinct. There are $m-d$ multisets of the form
\begin{align*}
\{C_d, C_1, \ldots, C_{m-1}\}, \{C_{d-1},C_d, C_1, \ldots, C_{m-2}\},\ldots, \{C_{2d-m+1},\ldots, C_d, C_1, \ldots, C_{d}\};
\end{align*}
there are $2d-m$ multisets of the form
\begin{align*}
\{C_{2d-m},\ldots, C_d, C_1, \ldots, C_{d-1}\}, \ldots, \{C_1,\ldots, C_d, C_1, \ldots, C_{m-d} \},
\end{align*}
and there are $m-d-1$ multisets of the form
\begin{align*}
\{C_m,C_1,\ldots, C_d, C_1, \ldots, C_{m-d-1} \}, \ldots, \{C_{d+2}, \ldots, C_m,C_1,\ldots, C_d, C_1 \}.
\end{align*}
The vector representation of these $(m-1)$ multisets are as follows.  
For the first group:
\begin{align*}
&(\underbrace{1,...,1}_\text{\it 2d-m},\underbrace{1,\ldots,1,2}_\text{\it m-d},\underbrace{1,\ldots,1,0}_\text{\it m-d}), (\underbrace{1,...,1}_\text{\it 2d-m},\underbrace{1,\ldots,1,2,2}_\text{\it m-d},\underbrace{1,\ldots,1,0,0}_\text{\it m-d}), \\
&\ldots \ldots, (\underbrace{1,...,1}_\text{\it 2d-m},\underbrace{2,\ldots,2}_\text{\it m-d},\underbrace{0,\ldots,0}_\text{\it m-d}).
\end{align*}
For the second group:
\begin{align*}
(\underbrace{1,...,1,2}_\text{\it 2d-m},\underbrace{2,\ldots,2,1}_\text{\it m-d},\underbrace{0,\ldots,0}_\text{\it m-d}), \ldots, (\underbrace{2,2,...,2}_\text{\it 2d-m},\underbrace{1,1,\ldots,1}_\text{\it m-d},\underbrace{0,\ldots,0}_\text{\it m-d}).
\end{align*}
For the third group:
\begin{align*}
(\underbrace{2,...,2,1}_\text{\it 2d-m},\underbrace{1,\ldots,1}_\text{\it m-d},\underbrace{0,\ldots,0,1}_\text{\it m-d}), \ldots, (\underbrace{2,1,...,1}_\text{\it d},0,1,\ldots,1).
\end{align*}
These vectors are distinct among each other and are not equal to ${\bf 1}_m$.  
\qed

\section{Proof of Proposition~\ref{prop:unitary code}}
\label{appendix:unitary code}

The if part is trivial.  
To establish the only if statement, suppose that $\mathcal{P}_1(\Phi(B({\bf x})))$ cannot uniquely determine $x_1$. 
This means that there exist ${\bf x}=(x_1,{\bf y})$ and ${\bf x'}=(x'_1,{\bf y}')$ with $\mathcal{P}_1(\Phi(B({\bf x})))=\mathcal{P}_1(\Phi(B({\bf x'})))$ and $x_1\ne x'_1$.

Let ${\bf z}=(x_1,{\bf y}')$ and consider the $\bf m$-block, $B(\bf z)$. 
We claim that $\mathcal{P}_1(\Phi(B({\bf x})))=\mathcal{P}_1(\Phi(B( {\bf z})))$.
To show this, note that $B({\bf x})$ consists of grid points of the form ${\bf u}=(x_1+a,{\bf y+b})$, where $0\leq a <\mi{1}$ and ${\bf 0}_{n-1}\leq {\bf b} <(\mi{2},\ldots,\mi{n})$.
Similarly, $B({\bf z})$ consists of grid points of the form ${\bf v}=(x_1+a,{\bf y'+b})$.
If $\bf u,v$ both belong to the sub-grid, $G^c_{\bf J}$, then
\[
\mathcal{P}_1(\Phi({\bf u}))=\mathcal{P}_1(\Phi({\bf v}))
=\Phi^{(1)}_{\bf J}(x_1+a).
\]
So $\mathcal{P}_1(\Phi(B({\bf x})))=\mathcal{P}_1(\Phi(B({\bf z})))$, and it follows from the assumption that $\mathcal{P}_1(\Phi(B({\bf z})))=\mathcal{P}_1(\Phi(B({\bf x}')))$. 

Consider two points ${\bf v}=(x_1+a,{\bf y'+b})\in B({\bf z})$ and ${\bf w}=(x_2+a,{\bf y'+b})\in B({\bf x'})$ that belong to the same sub-grid $\bf J$.  
Recall that each $\bf m$-block contains only one point in sub-grid $\bf J$.  
Since the color sets for different sub-grids are distinct, $\mathcal{P}_1(\Phi(B({\bf z})))$ and $\mathcal{P}_1(\Phi(B({\bf x}')))$ are identical, which implies that
\[
\mathcal{P}_1(\Phi({\bf w}))=\Phi^{(1)}_{\bf J}(x_2+a)=
\mathcal{P}_1({\Phi(\bf v}))
=\Phi^{(1)}_{\bf J}(x_1+a).
\]
For $2\leq i \leq n$, $\mathcal{P}_i(\Phi({\bf w}))=\mathcal{P}_i({\Phi(\bf v}))$ by definition of $\Phi$.
Hence, $\Phi({\bf v})=\Phi({\bf w})$. 
This holds for all similar pairs of grid points in $B(\bf x')$ and $B(\bf z)$, so $\Phi(B({\bf x}'))=\Phi(B({\bf z}))$.  
But the two $m$-blocks are distinct, a contradiction.  
Similar arguments can be applied to coordinate $l> 1$ to complete the proof of the only if statement.
\qed


\bibliographystyle{IEEEtran}

\bibliography{IEEEabrv,reference,reference_vlp}
\end{document}